\documentclass[%
11pt,
onecolumn,
tightenlines,
superscriptaddress,
notitlepage,
preprintnumbers,
nofootinbib,
amsmath,amssymb,amsthm,
aps,
eqsecnum
]{revtex4-2}

\usepackage[utf8]{inputenc}
\usepackage{mathtools}
\usepackage{amsmath}
\usepackage{amsthm}
\usepackage{amsbsy}
\usepackage{amssymb}
\usepackage{amscd}
\usepackage{amsfonts}
\usepackage{bm}
\usepackage{graphics}
\usepackage{graphicx}
\usepackage[plain]{fancyref}
\usepackage{soul}
\usepackage{hyperref}
\setcitestyle{numbers,sort&compress}
\usepackage{verbatim}
\usepackage[margin=25mm]{geometry}
\usepackage{paralist}
\usepackage{wasysym}
\usepackage{xcolor}
\usepackage{lineno}

\graphicspath{ {figs/} }

\newcommand*{\fancyrefapplabelprefix}{app}
\frefformat{plain}{\fancyrefapplabelprefix}{appendix~#1}
\Frefformat{plain}{\fancyrefapplabelprefix}{Appendix~#1}
\frefformat{main}{\fancyrefapplabelprefix}{appendix~#1}
\Frefformat{main}{\fancyrefapplabelprefix}{Appendix~#1}

\newcommand{\partialInline}[2]{\partial #1 / \partial #2}

\newcommand{\derivN}[3]{\frac{\mathrm{d}^{#3} #1}{\mathrm{d} #2^{#3}}}

\newcommand{\rmag}{r} 
\newcommand{\vvec}[1]{\bm{#1}}
\newcommand{\rvec}{\vvec{\rmag}}
\newcommand{\dir}[1]{\hat{#1}}
\newcommand{\rdir}{\dir{\rvec}}
\newcommand{\tens}[1]{\bm{#1}}

\newcommand{\cLen}{\ell_0}
\newcommand{\mLen}{b} 
\newcommand{\n}{n} 
\newcommand{\chainStretch}{\gamma} 
\newcommand{\relChainStretch}{\gamma_r} 
\newcommand{\chainDensity}{N} 

\newcommand{\nvec}{\hat{\vvec{n}}} 
\newcommand{\cStrSqr}{\zeta}

\newcommand{\eFieldMag}{E} 
\newcommand{\eFieldVec}{\vvec{\eFieldMag}}
\newcommand{\eDir}{\dir{\eFieldVec}}
\newcommand{\dipoleMag}{\mu} 
\newcommand{\dipoleVec}{\vvec{\dipoleMag}}

\newcommand{\um}{u}
\newcommand{\dipole}{\dipoleVec} 
\newcommand{\susceptibility}{\chi}
\newcommand{\susPara}{\susceptibility_\parallel}
\newcommand{\susPerp}{\susceptibility_\perp}
\newcommand{\dsus}{\Delta \susceptibility}
\newcommand{\dipoleSusceptibility}{\tens{\susceptibility}}
\newcommand{\sustensexpr}{
	\susPara \nvec \otimes \nvec + \susPerp \left(\identity - \nvec \otimes \nvec\right)
}

\newcommand{\trans}[1]{#1^{T}} 

\newcommand{\kB}{k} 
\newcommand{\T}{T} 
\newcommand{\chainFreeEnergy}{w}
\newcommand{\WLCFreeEnergy}{\chainFreeEnergy_{\text{wlc}}}
\newcommand{\KGFreeEnergy}{\chainFreeEnergy_{\text{KG}}}
\newcommand{\GaussFreeEnergy}{\chainFreeEnergy_{\text{G}}}
\newcommand{\DEFreeEnergy}{\chainFreeEnergy_{\text{DE}}}
\newcommand{\bioFreeEnergy}{\chainFreeEnergy_{\text{bio}}}
\newcommand{\forcemag}{f}
\newcommand{\forcevec}{\vvec{\forcemag}}
\newcommand{\pFrame}{\vvec{P}}
\newcommand{\pDir}{\hat{\vvec{v}}}
\newcommand{\uVec}{\dir{\vvec{\upsilon}}}
\newcommand{\nuVec}{\dir{\vvec{\nu}}}
\newcommand{\probDens}{\rho}
\newcommand{\elecFreeEnergy}{\chainFreeEnergy_{\text{e}}}
\newcommand{\orientFreeEnergy}{\chainFreeEnergy_{\text{o}}}
\newcommand{\unodim}{\kappa}
\newcommand{\uOnodim}{\kappa_{\perp}}
\DeclareMathOperator{\erf}{erf}
\newcommand{\erfw}{\erf\left(\sqrt{\unodim}\right)}
\newcommand{\vperm}{\epsilon_0}

\newcommand{\unitSphere}{\mathbb{S}^2}
\newcommand{\rSphere}[1]{\unitSphere\left({#1}\right)}
\newcommand{\df}[1]{\text{d}#1} 

\newcommand{\Lang}{\mathcal{L}}
\newcommand{\invLang}{\Lang^{-1}}

\newcommand{\pLen}{p}
\newcommand{\zmultzero}{\invLang\left(\chainStretch\right)}

\newcommand{\numChains}{M}
\newcommand{\defMap}{\bm{\Phi}}
\newcommand{\pullback}[1]{\tilde{#1}}
\newcommand{\Rmag}{\pullback{\rmag}}
\newcommand{\Rvec}{\pullback{\rvec}}
\newcommand{\Rdir}{\uVec}
\newcommand{\avg}[1]{\left< #1 \right>}
\newcommand{\avgOverR}[1]{\avg{#1}}

\newcommand{\generic}{\Box}
\newcommand{\freeEnergyDensity}{\mathcal{W}}
\newcommand{\innerFreeEnergyDensity}{\widehat{\freeEnergyDensity}}
\newcommand{\netFreeEnergyDensity}{\freeEnergyDensity_{\text{net}}}

\newcommand{\F}{\tens{F}}

\newcommand{\cGreenSym}{C}
\newcommand{\cGreen}{\tens{\cGreenSym}}
\newcommand{\strTens}{\tens{V}}
\newcommand{\polRot}{\tens{R}}
\newcommand{\genRot}{\tens{Q}}
\newcommand{\genRotStar}{\genRot^*}
\newcommand{\changeCoord}[1]{\bar{#1}}
\newcommand{\flucField}{f}
\newcommand{\rodVec}{\vvec{\omega}}
\newcommand{\eulerx}{\alpha}
\newcommand{\eulery}{\beta}
\newcommand{\eulerz}{\xi}
\newcommand{\chainStretchStar}{\chainStretch^*}

\newcommand{\SOThree}{SO\left(3\right)}

\newcommand{\Gradd}{\mathrm{Grad}}

\newcommand{\takeGrad}[1]{\Gradd \: #1}
\newcommand{\gapvec}{\bm{\mathcal{E}}}

\newcommand{\xj}{x}
\newcommand{\yj}{y}
\newcommand{\xvec}{\vvec{\xj}} 
\newcommand{\yvec}{\vvec{\yj}} 
\DeclareMathOperator{\diag}{diag} 
\DeclareMathOperator{\trace}{Tr} 
\newcommand{\pStretchSymbol}{\lambda} 
\newcommand{\pStretch}[1]{\pStretchSymbol_{#1}} 

\newcommand{\identity}{\tens{I}}
\newcommand{\euclid}[1]{\hat{\vvec{e}}_{#1}}
\newcommand{\shear}{s}

\newcommand{\orderOf}[1]{\mathcal{O}\left(#1\right)}

\DeclareMathOperator{\Tr}{Tr}

\newcommand{\smallpar}{\epsilon}

\newcommand{\Reals}{\mathbb{R}}

\newcommand{\shearDef}{s}
\newcommand{\tmag}{t}
\newcommand{\tvec}{\vvec{\tmag}}
\newcommand{\tnodim}{\tau}
\newcommand{\vol}{\mathcal{V}}
\newcommand{\polar}{\theta}
\makeatletter
\newcommand*{\gnuplotinput}[2][1.0]{%
  \begingroup
  \let\@gnplt@input@includegraphics=\includegraphics
  \def\includegraphics##1{\@gnplt@input@includegraphics[scale=#1]{#2}}%
  \let\@gnplt@input@picture=\picture
  \def\picture{\unitlength=#1\unitlength\relax\@gnplt@input@picture}%
  \input{#2}%
  \endgroup
}
\makeatother

\newtheorem{proposition}{Proposition}
\newtheorem*{lemma}{Lemma}

\theoremstyle{remark}
\newtheorem{remark}{Remark}

\newcommand*{\fancyrefproplabelprefix}{prop}
\frefformat{plain}{\fancyrefproplabelprefix}{proposition~#1}
\Frefformat{plain}{\fancyrefproplabelprefix}{Proposition~#1}
\frefformat{main}{\fancyrefproplabelprefix}{proposition~#1}
\Frefformat{main}{\fancyrefproplabelprefix}{Proposition~#1}

\newcommand*{\fancyreflemlabelprefix}{lem}
\frefformat{plain}{\fancyreflemlabelprefix}{lemma~#1}
\Frefformat{plain}{\fancyreflemlabelprefix}{Lemma~#1}
\frefformat{main}{\fancyreflemlabelprefix}{lemma~#1}
\Frefformat{main}{\fancyreflemlabelprefix}{Lemma~#1}

\newcommand{\singleGraphWidth}{0.75\linewidth}
\newcommand{\capTitle}[1]{\emph{#1}}

\renewcommand{\hl}[1]{{#1}}
\newenvironment{hlbreakable}%
{}%
{}

\begin{document}

\title{Polymer networks which locally rotate to accommodate stresses, torques, and deformation}

\author{Matthew Grasinger}
\email{matthew.grasinger.1@us.af.mil}
\affiliation{Materials and Manufacturing Directorate, Air Force Research Laboratory}
\preprint{To appear in Journal of the Mechanics and Physics of Solids, doi: \url{https://doi.org/10.1016/j.jmps.2023.105289}.}



\begin{abstract}
  Polymer network models construct the constitutive relationships of a broader polymer network from the behavior of a single polymer chain (e.g. viscoelastic response to applied forces, applied electromagnetic fields, etc.). 
  Network models have been used for multiscale phenomena in a variety of contexts such as rubber elasticity, soft multifunctional materials, biological materials, and even the curing of polymers. 
  For decades, a myriad of polymer network models have been developed with differing numbers of chains, arranged in different ways, and with differing symmetries.
  To complicate matters further, there are also competing assumptions for how macroscopic variables (e.g. deformation) are related to individual chains within the network model.
  In this work, we propose a simple, intuitive assumption for how the network locally rotates relative to applied loading (e.g. stresses, external fields) and show that this assumption unifies many of the disparate polymer network models--while also recovering one of the most successful models for rubber elasticity, the Arruda-Boyce $8$-chain model. 
  The new assumption is then shown to make more intuitive predictions (than prior models) for stimuli-responsive networks with orientational energies (e.g. electroactive polymers), which is significant for shape morphing and designing high degree of freedom actuators (e.g. for soft robotics).
  Lastly, we unveil some surprising consequences of the new model for the phases of multistable biopolymer and semi-crystalline networks.
\end{abstract}

\maketitle

\section{Introduction} \label{sec:intro}

\hl{
Constitutive models for (solid) polymer networks can roughly be divided into two categories: \begin{inparaenum}[1)] \item phenomenological models based on strain invariants and/or principal stretches, and \item micromechanically based models that aim to connect macromolecular and network structural properties to the bulk, continuum scale. \end{inparaenum}
In the first category, the well known Neo-Hookean is among the simplest.
Other well known examples include Mooney-Rivlin type models~\cite{mooney1940theory,rivlin1948large} and Ogden models~\cite{destrade2022ogden,ogden1972large}.
The focus of this work will be on the latter category: micromechanically based models.
}

Entropy drives the elasticity of soft networks consisting of long polymer chains.
As a result, there is a rich history of statistical mechanics guiding our understanding of polymer networks~\cite{treloar1975physics}.
Typically a continuum model derived from such principles starts at the scale of a polymer chain where the free energy of a polymer, as a function of its stretch, is determined.
Then the continuum response to deformation can be constructed by using the single chain response in a so-called polymer network model.
These network models allow one to relate macroscopic state variables, such as deformation, to individual chains within the network.

Many polymer network models exist in the literature.
Some have a discrete number of chains arranged in a representative volume element (RVE) (i.e. unit cell) while others model the network via a probability density of polymer chains.
There are also competing assumptions for how macroscopic deformation is related to the deformations of chains within the network.
Historically, for discrete network models, there are the $3$-chain~\cite{james1943theory}, the $4$-chain~\cite{treloar1943elasticity} \hl{(i.e. Flory-Rehner model)}, and the $8$-chain~\cite{arruda1993threee} models.
Continuous network models began with the full network model of Wu and van der Giessen~\cite{wu1993improved} which assumes that all of chains in the network are of the same length (in the reference configuration) and that the directions of chains are uniformly distributed over the unit sphere.
In response to the (at times) underwhelming fit of the full network model to benchmark data in rubber elasticity, it was improved upon by the microsphere model of~\citet{miehe2004micro}.
A key feature of the microsphere model was the development of a new assumption for how macroscopic deformation is related to chains within the network. 
The assumption involved a so-called ``stretch fluctuation field'' on the unit sphere of chain directions.
The stretch fluctuation field is determined by minimizing the average free energy of the network subject to homogenization-based constraints.
For certain fitted parameters, the microsphere model recovers the $8$-chain model exactly.
\hl{Around the same time~\citet{beatty2003average} suggested a different perspective on the full network model: it was shown that if one takes the chain free energy for the average stretch of the full network model--as opposed to the average chain free energy--one again recovers the $8$-chain model.
	
Since these earlier developments, several variants have been proposed that aim to retain much of the simplicity and physical-basis of earlier models while also providing a better fit to the S-shaped stress-strain curves of the well-known Treloar data for rubber elasticity~\cite{erman1980moments,bechir2010three,miroshnychenko2009heuristic,xiang2018general,kroon20118,zhan2023new}. 
Many of these variants decompose the free energy density of the network additively into a contribution which represents the ``cross-linked network''--often modelled using the $8$-chain or a full network model--and a topological constraint contribution due to chain entanglements within the network.
The topological constraint contribution has taken different forms: some based on chain statistics and micromechanics~\cite{erman1980moments,xiang2018general,kroon20118} while others appear more phenomenological~\cite{bechir2010three,miroshnychenko2009heuristic}. 
For example,~\citet{bechir2010three} argues that the standard $3$-chain model can, phenomenologically speaking, be used to model the topological constraint contribution.
Whereas, alternatively,~\citet{miroshnychenko2009heuristic} models the topological contribution using strain invariants.
For full network type models, a new macro-to-micro kinematic assumption has recently been developed, based on an affine stretch projection and a microscale analog of the Biot stress, which performs well for capturing complex and multiaxial stress-strain relationships~\cite{zhan2023new}.
Many constitutive models for rubber-like elasticity exist in the literature that attempt to balance between competing objectives: \begin{inparaenum}[1)] \item contain as few fitting parameters as possible, \item have the ability to reproduce complex deformation behavior, and \item in the interest of informing material design, have model parameters which are physically interpretable and can be reasonably connected to the underlying microstructure of the material. \end{inparaenum}
Comprehensive reviews can be found in~\citet{hossain2015eight} and~\citet{dal2021performance}.
}

Polymer network models have been applied to a vast number of different polymer networks related to adhesives~\cite{zhao2022network}, biomechanics~\cite{grekas2021cells,song2022hyperelastic,kuhl2005remodeling,alastrue2009anisotropic}, magneto-active polymers (e.g. hard-magnetic soft materials)~\cite{garcia2021microstructural,moreno2022effects} and electro-active polymers~\cite{grasingerIPtorque,grasingerIPflexoelectricity,grasinger2019multiscale,cohen2016electroelasticity,cohen2018generalized,itskov2018electroelasticity}.
Network models have also been used for capturing phenomena such as viscoelasticity~\cite{bergstrom1998constitutive,zhao2022network}, the curing of networks~\cite{hossain2011modelling}, excluded volume induced strain-hardening and incompressibility~\cite{khandagale2022statistical}, and chain scission and fracture~\cite{mulderrig2021affine,mao2017rupture}.
As a result, continuous network models have been generalized to anisotropic chain distributions~\cite{alastrue2009anisotropic,grasinger2020architected} and polydisperse networks (i.e. networks consisting of chains of different lengths)~\cite{mulderrig2021affine}.
While it is widely accepted that the $8$-chain and microsphere models \hl{(and their variants)} are \hl{among} the most accurate for the elasticity of rubber-like materials, the appropriate polymer network model for stimuli-responsive~\cite{grasingerIPtorque,grasingerIPflexoelectricity,grasinger2019multiscale,garcia2021microstructural,moreno2022effects}, biopolymer~\cite{song2022hyperelastic,kuhl2005remodeling,alastrue2009anisotropic}, and semi-crystalline networks~\cite{rastak2018non} are still active areas of research; as is the appropriate polymer network model for fracture~\cite{mulderrig2021affine,mao2017rupture,lei2022network}.
Recent results suggest, for instance, that the $8$-chain and microsphere models are inaccurate and make nonphysical predictions for some biopolymer networks~\cite{song2022hyperelastic}.
Similarly, theoretical predictions made by the $3$- and $8$-chain model for certain stimuli-responsive polymer networks, where chains have a free energy dependence on their orientation (relative to an applied electromagnetic field), can also be counter-intuitive~\cite{grasingerIPtorque,grasinger2019multiscale}.
Meanwhile, properly understanding biopolymer and stimuli-responsive networks will be key to achieving breakthroughs in soft, biologically-inspired robotics\hl{~\cite{cianchetti2018biomedical,majidi2019soft}}, advanced prosthetics, wearable sensors\hl{~\cite{ates2022end,heikenfeld2018wearable}}, and the interfacing of biology and electronics, more broadly.


The primary focus of this work is on discrete network models.
Numerical integration of the full-network and microsphere models is nontrivial~\cite{verron2015questioning}. 
More importantly, discrete network models can characterize the geometrical features of the network which are significant for the homogenization of composites (e.g.~\cite{castaneda2011homogenization}), strain gradient elasticity~\cite{jiang2022strain}, and strain gradient couplings such as flexoelectricity~\cite{grasingerIPflexoelectricity}.
Finally, it will later be shown that, in certain contexts, discrete polymer network models have an apparent advantage over their continuous counterparts because cross-links (i.e. junctions between chains) in real polymer networks consist of a finite number of chains.

The standard macro-to-micro kinematic assumption for discrete polymer network models--motivated by coordinate system invariance--is that the RVE rotates such that its edges lie along the principal directions of stretch.
However, recent theoretical work and experimental observations suggest the importance of the local rotation of polymer network RVEs~\cite{kuhl2005remodeling,grasingerIPtorque,cohen2018generalized}.
In \fref{sec:polymer-chains}, as a preamble, some broadly used polymer chain models are reviewed.
Then, in \fref{sec:new-model}, we develop a new macro-to-micro kinematic assumption for discrete polymer network models where the rotation of the RVE is chosen such that it has minimal free energy upon deformation (i.e. the network locally rotates to most efficiently accommodate stresses, torques, and deformation).
In \fref{sec:unification}, it is shown that, for the broadly used polymer chain models (e.g. polymer chain models for rubber elasticity), the new kinematic assumption recovers the successful Arruda-Boyce $8$-chain and microsphere models--for all of the $3$-, $4$-, $6$-, and $8$-chain RVEs.
Thus, the new assumption unifies disparate polymer network models and reproduces benchmark experimental data.
Then in \fref{sec:stimuli-responsive} it is shown that the new macro-to-micro assumption produces discrete polymer network models which are more thermodynamically consistent and make more intuitive predictions for stimuli-responsive networks.
In \fref{sec:nonconvex}, surprising implications of the new model for biopolymer and semi-crystalline networks are unveiled.
\Fref{sec:conclusion} concludes the work.

\section{Polymer chains} \label{sec:polymer-chains}

The fundamental kinematic description of a polymer chain is its the end-to-end vector, $\rvec$, which is the vector that spans from the start of the chain to its end.
Let $\rmag = \left|\rvec\right|$ denote its magnitude.
Many polymers have a well defined maximum length, $\cLen$, such that $\rmag \in \left[0, \cLen\right]$.
We refer to $\chainStretch = \rmag / \cLen$ as the chain stretch.
A natural way to model the elastic response of a chain is through the specification of its free energy as a function of its stretch, i.e. $\chainFreeEnergy = \chainFreeEnergy\left(\chainStretch\right)$.
Here we briefly review $3$ examples of free energy functions: the Gaussian chain (GC), the Kuhn and Gr\"{u}n chain (KGC), and the wormlike chain (WLC).

In the freely jointed chain model, it is assumed that the polymer is made up of $\n$ rigid links of length $\mLen$--called Kuhn segments--so that $\cLen = \n \mLen$.
The links are bonded end-to-end and links are free to rotate about their neighboring bonds.
Given the model assumptions, the free energy response can be \hl{derived} using statistical mechanics; however the exact solution is untenable.
As a result, approximations are derived in certain limits or by introducing additional simplifying assumptions.
One such approximation for the free energy of a freely jointed chain is the Gaussian chain approximation:
\begin{equation} \label{eq:GC}
	\GaussFreeEnergy\left(\chainStretch\right) = \frac{3}{2} \n \kB \T \chainStretch^2.
\end{equation}
The GC approximation is derived using random walk statistics and is valid for small to moderate chain stretches, but does not capture the finite extensibility of the chain.
A better approximation of freely jointed chain behavior is obtained by the KGC formula~\cite{treloar1975physics}:
\begin{equation} \label{eq:KGC}
	\KGFreeEnergy\left(\chainStretch\right) = \n \kB \T \left(\chainStretch \invLang\left(\chainStretch\right) + \log \left(\frac{\invLang\left(\chainStretch\right)}{4\pi \sinh \invLang\left(\chainStretch\right)}\right) \right).
\end{equation}
which does capture the finite extensibility of the chain (i.e. $\KGFreeEnergy \rightarrow \infty$ as $\chainStretch \rightarrow 1$).

An alternative that describes many realistic and stiffer polymers is the wormlike chain model.
In contrast to the FJC, which consists of a discrete number of rigid links, the WLC is a kind of continuously flexible rod.
The bending stiffness of the chain per unit length is given by the persistence length, $\pLen$.
A simple approximation for its free energy is~\cite{kuhl2005remodeling}:
\begin{equation} \label{eq:WLC}
	\WLCFreeEnergy\left(\chainStretch\right) = \frac{\kB \T \cLen}{4 \pLen}\left(2 \chainStretch^2 + \hl{\frac{1}{1 - \chainStretch}} - \chainStretch\right),
\end{equation}
which also captures the finite extensibility of the chain.
The WLC has been used to successfully model DNA and other biomolecules~\cite{kuhl2005remodeling,alastrue2009anisotropic}.

Given the free energy response of a single chain, its elastic response can be obtained directly by minimization:
\begin{equation}
	\min_{\rvec} \left\{\chainFreeEnergy - \forcevec \cdot \rvec\right\} \implies \forcevec = -\frac{\partial \chainFreeEnergy}{\partial \rvec}.
\end{equation}

\section{New macro-to-micro kinematic assumptions} \label{sec:new-model}

\subsection{Continuum mechanics}

To set notation, we introduce some fundamental concepts in continuum solid mechanics.
Let $\xvec$ be a material point of a solid body in the reference configuration and $\yvec = \defMap\left(\xvec\right)$ its corresponding point in the deformed configuration.
Here $\defMap$ is the mapping which describes the deformation of the solid body.
The deformation gradient, $\F = \takeGrad{\defMap}$, maps so-called line elements--which are infinitesimal changes in position, $\delta \xvec$--from the reference to the current configuration, i.e. $\delta \yvec = \F \delta \xvec$.
The right Cauchy-Green tensor, $\cGreen = \trans{\F} \F$, describes the stretches of line elements under $\F$.
Physically, we require that $\det \F > 0$.
Then by the polar decomposition theorem $\F$ can be uniquely decomposed into $\F = \polRot \strTens$ where $\polRot \in \SOThree$, $\strTens = \sqrt{\cGreen}$ is positive definite, and $\SOThree$ is the group of $3$ dimensional rotations.
Moving forward, we refer to $\strTens$ as the stretch tensor and its eigenvalues and eigenvectors as the principal stretches, $\pStretch{i}$, and principal directions, $\pDir_i$, respectively.
Because $\strTens$ is symmetric, the principal directions can be chosen such that $\pDir_i \cdot \pDir_j = \delta_{ij}$ and this set of directions constitutes a principal frame.
Analogous to the single chain case, the elasticity of the solid polymer network can be modeled via a free energy density $\freeEnergyDensity = \freeEnergyDensity\left(\F\right)$.

\subsection{Classical polymer network models}

The free energy density of a polymer network can be constructed from the behavior of a single chain through the use of a polymer network model.
Discrete polymer network models consist of a representative volume element (RVE), polymer chains arranged within the RVE, and associated assumptions for how $\F$ acts on the RVE.
\Fref{fig:classical-network-models} depicts examples of classical polymer network models: \begin{inparaenum}[a)] \item the $3$-chain model~\cite{james1943theory}, \item the $4$-chain model~\cite{treloar1943elasticity}, \item the $8$-chain model~\cite{arruda1993threee}, and \item the full network model~\cite{wu1993improved}. \end{inparaenum}
A key consideration for the assumed relationship between $\F$ and the deformation of an RVE is that the resulting constitutive model be coordinate system invariant.
The $3$-chain and $8$-chain models handle this via the same kinematic assumption: \begin{inparaenum}[1)] \item the RVE rotates such that the edges of the cube stretch in the principal frame and \item the edges of the cube stretch according to the principal stretches. \end{inparaenum}
We refer to this as the \emph{principal frame assumption}.
The principal frame assumption can be formulated as follows.
Let $\Rvec_i, i = 1, 2, \dots, \numChains$ denote the chain end-to-end vectors before deformation (where $\numChains$ is the number of chains in the RVE).
Then
\begin{equation}
	\rvec_i = \strTens \pFrame \Rvec_i
\end{equation}
where
\begin{equation}
	\pFrame = \begin{pmatrix}
		| & | & | \\
		\pDir_1 & \pDir_2 & \pDir_3 \\
		| & | & |
	\end{pmatrix}
\end{equation}
is an orthogonal transformation which aligns the edges of the RVE with the principal frame.
The stretches satisfy
\begin{equation} \label{eq:chain-stretch}
	\chainStretch_i = \frac{1}{\cLen} \left|\strTens \pFrame \Rvec_i\right| = \frac{1}{\cLen} \sqrt{\Rvec_i \cdot \changeCoord{\cGreen} \Rvec_i}
\end{equation}
where by $\changeCoord{\cGreen}$ we mean a similarity transformation of $\cGreen$.
In this case, $\changeCoord{\cGreen}$ is a diagonalization of $\cGreen$.
Assuming the interactions between chains are negligible, the free energy density is approximated by an average of the chain free energies over the RVE times the number of chains per unit volume, $\chainDensity$.
Let $\avgOverR{ \generic }$ denote the quantity $\generic$ averaged over the chains in the RVE.
The $3$-chain polymer network model results in the constitutive relation
\begin{equation} \label{eq:3-chain}
	\freeEnergyDensity\left(\F\right) = \chainDensity \avgOverR{\chainFreeEnergy\left(\frac{\left|\strTens \pFrame \Rvec\right|}{\cLen}\right)} = \frac{\chainDensity}{3} \sum_{i=1}^{3} \chainFreeEnergy\left(\frac{\Rmag}{\cLen} \pStretch{i}\right);
\end{equation}
and the $8$-chain network model results in a different constitutive relation
\begin{equation} \label{eq:8-chain}
	\freeEnergyDensity\left(\F\right) = \chainDensity \chainFreeEnergy\left(\frac{\Rmag}{\cLen} \sqrt{\frac{\pStretch{1}^2+\pStretch{2}^2+\pStretch{3}^2}{3}}\right).
\end{equation}
\begin{figure*}
	\centering
	\includegraphics[width=\linewidth]{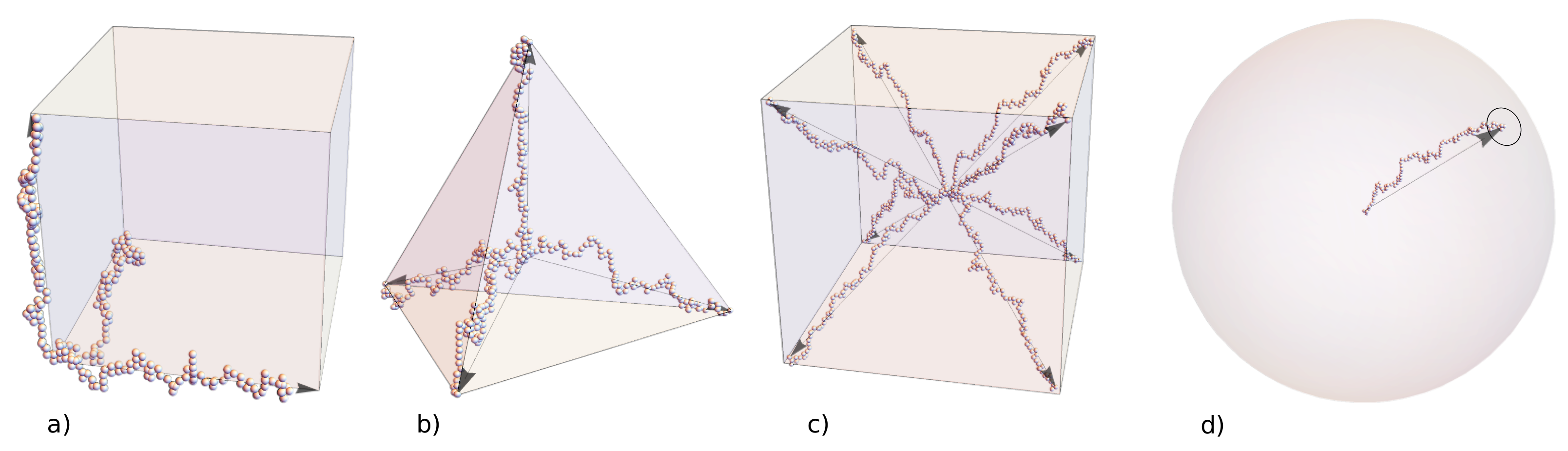}
	\caption{
		\capTitle{Classical polymer network models.} The representative volume elements of the a) $3$-chain model, b) the $4$-chain model, and c) the $8$-chain model. The sphere in d) is representative of a continuous, uniform density of chain directions constituent of the full network and microsphere models.
	}
	\label{fig:classical-network-models}
\end{figure*}

It seems to be widely accepted that the $8$-chain model produces more accurate constitutive models for a wide array of real polymer networks~\cite{miehe2004micro,arruda1993threee,kuhl2005remodeling}.
It is noteworthy then that \begin{inparaenum}[1)] \item there is such a stark difference between the $3$-chain and $8$-chain models and \item that it is so widely accepted that the $8$-chain produces better constitutive models. \end{inparaenum}
One may ask: does this suggest that more real polymer networks consist of cross-links of $8$ chains as opposed to $3$, $4$, or $6$?
If not, what is it that drives the success of the $8$-chain model?

Further questions arise in the context of polymer networks with orientational energies.
For example, if the monomers in a chain have electrical dipoles and an electric field is externally applied, then the chain free energy is not only dependent on the magnitude of stretch, but also the direction of stretch, $\rdir$~\cite{grasinger2022statistical,grasinger2020statistical,grasingerIPflexoelectricity}.
The $8$-chain model with the principal frame assumption can lead to counter-intuitive, and likely nonphysical, predictions in such cases~\cite{grasingerIPtorque,grasinger2019multiscale}.

\subsection{Free rotation assumption}

Here we propose a new set of macro-to-micro kinematic assumptions for discrete polymer networks with an emphasis on correctly capturing the orientation of chains within the network after some local macroscopic deformation $\F$.
Because we are interested in stimuli-responsive networks with orientational energies, we assume from here on that the chain free energy is given as a function of the end-to-end vector, i.e. $\chainFreeEnergy = \chainFreeEnergy\left(\rvec\right)$.
The principal frame assumption ensures that the resulting constitutive relation is independent of a particular choice of coordinate system; however, there does not appear to be strong physical justification beyond this and the assumption can breakdown when chains have an orientational energy.
To address this, consider the notion of locally rotating a cross-link and its associated RVE by some arbitrary $\genRot \in \SOThree$.
According to the principal frame assumption: $\genRot = \pFrame$.
Here, however, we instead \emph{propose to choose $\genRot$ such that the free energy density is minimized}.
The principal argument is that the polymer network is able to locally, microscopically rotate in order to accommodate the macroscopic deformation in the most energetically efficient way possible.

Next we consider the deformation of the RVE via the stretch tensor, $\strTens$.
Although this provides a clean geometric picture, it also seems to lack a strong physical justification.
Instead we recall that line elements get mapped from the reference configuration to the deformed configuration by the deformation gradient, $\F$.
We posit that polymer end-to-end vectors are sufficiently ``short'' so as to be considered line elements and assume that, after rotation by $\genRot$, end-to-end vectors are deformed according to $\F$.

A final consideration is the possibility of a nonlocal interaction energy associated with rotating the cross-link relative to the broader network, $\netFreeEnergyDensity$.
This rotation, for instance, may cause stresses and deformations on neighboring cross-links.
Considering this contribution, the new free energy density is formulated as follows:
\begin{subequations} \label{eq:new-model}
\begin{align}
	\label{eq:new-model-min-rot}
	\freeEnergyDensity\left(\F\right) &= \min_{\genRot \in \SOThree} \left\{\innerFreeEnergyDensity\left(\F, \genRot\right) + \netFreeEnergyDensity\left(\genRot, \dots\right)\right\}, \\
	\label{eq:new-model-inner}
	\innerFreeEnergyDensity\left(\F, \genRot\right) &= \chainDensity \avgOverR{\chainFreeEnergy\left(\F \genRot \Rvec\right)}.
\end{align}
\end{subequations}
The notation $\netFreeEnergyDensity\left(\genRot, \dots\right)$ is used to emphasize that, because this contribution to the free energy is nonlocal, it is possibly a function of $\F$, $\takeGrad{\F}$, and/or $\takeGrad{\genRot}$.
As such, the form of $\netFreeEnergyDensity\left(\genRot, \dots\right)$ may be important when modelling strain gradient effects such as strain gradient elasticity~\cite{jiang2022strain} or flexoelectricity~\cite{grasingerIPflexoelectricity}.
These effects will be a topic of future work.
\hl{In the interim, the issue of compatibility when locally rotating networks are subject to an inhomogeneous deformation, and its potential implications for $\netFreeEnergyDensity$, is briefly explored in \fref{app:compatibility}.}
However, throughout the remainder of the work, motivated by both simplicity and the success of the $8$-chain model, we will assume that the $\netFreeEnergyDensity$ contribution is negligible and will not consider it further.
The newly developed constitutive model--depicted in \fref{fig:free-rotation}--then reduces to a polymer network which is able to locally rotate in order to accommodate the macroscopic deformation, $\F$, according to what has been referred to in previous work as \emph{the principle of minimum average free energy}~\cite{miehe2004micro,rastak2018non,mulderrig2021affine}.
For conciseness, we refer to this new model assumption as the \emph{free rotation assumption}.
\begin{figure*}
	\centering
	\includegraphics[width=0.75\linewidth]{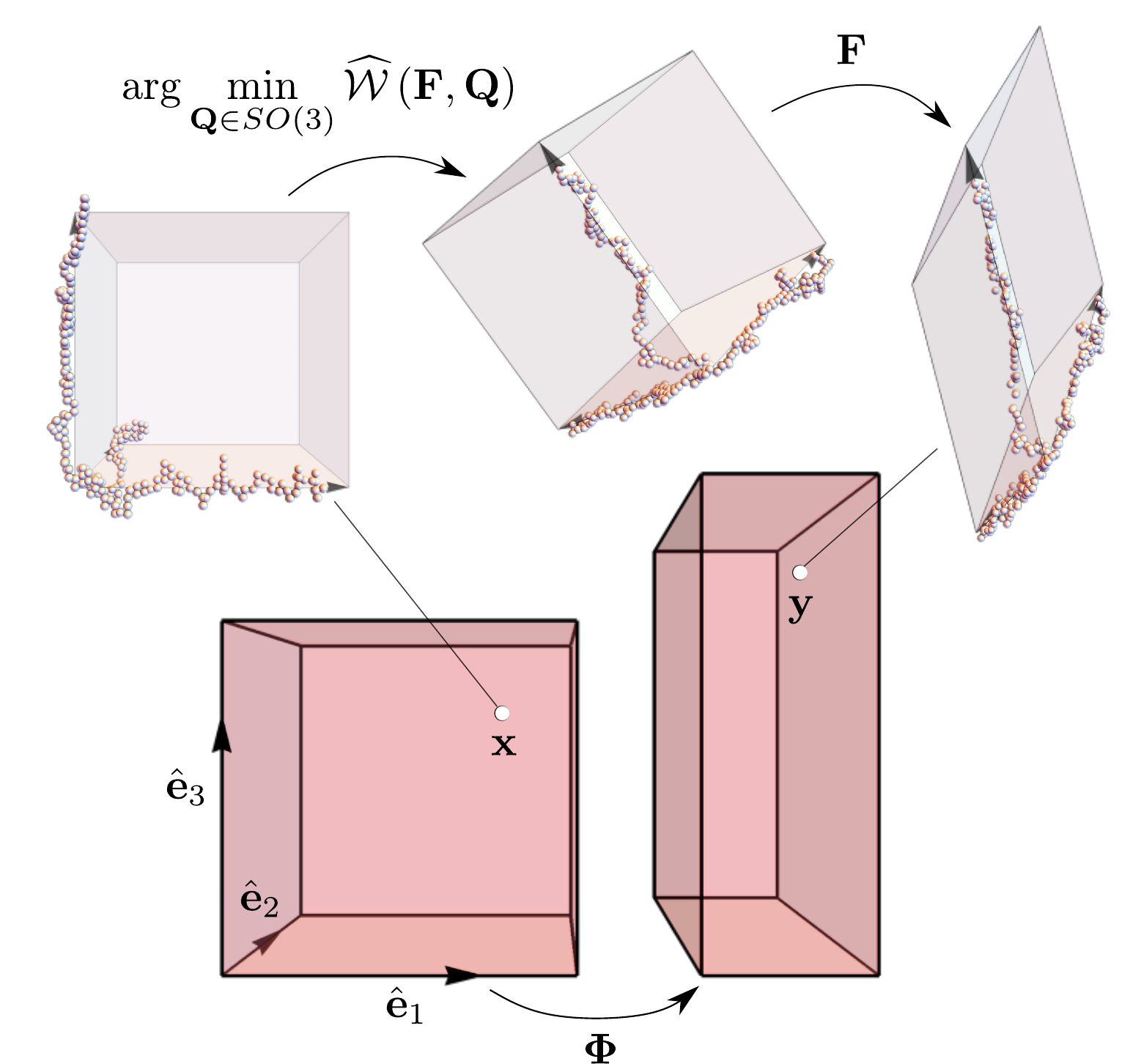}
	\caption{
		\capTitle{Free rotation assumption.}
		A solid polymer network (bottom row) is deformed according to the mapping $\defMap$.
		The RVE of the network at material point $\xvec$, in the reference configuration, is shown in the top left.
		In the deformed configuration, the RVE, now at $\yvec = \defMap\left(\xvec\right)$, first locally rotates to most efficiently accommodate the macroscopic deformation $\F$ (i.e. rotates by $\arg \min_{\genRot \in \SOThree} \innerFreeEnergyDensity\left(\F, \genRot\right)$), then is deformed according to $\F$.
	}
	\label{fig:free-rotation}
\end{figure*}

\subsection{$\SOThree$ representations and optimization}

It is well known that there are numerical pathologies associated with the Euler angle representation when optimizing over rotations~\cite{kuehnel2003minimization}.
In this work, we instead use the exponential, or ``axis-angle'', representation.
The Rodrigues vector, $\rodVec \in \Reals^3$, describes the angle of rotation, $\varphi = \left|\rodVec\right|$ and an axis about which to rotate $\dir{\vvec{u}} = \rodVec / \varphi$.
The rotation, $\genRot$, is obtained by taking the exponential of the generating skew-symmetric tensor
\begin{equation}
	\tens{A} = \begin{pmatrix}
			0 & -\dir{u}_3 & \dir{u}_2 \\
			\dir{u}_3 & 0 & -\dir{u}_1 \\
			-\dir{u}_2 & \dir{u}_1 & 0
		\end{pmatrix}; \quad \genRot = \exp\left(\varphi \tens{A}\right) = \identity + \sin \varphi \tens{A} + \left(1 - \cos \varphi\right) \tens{A}^2.
\end{equation}
For the numerical examples presented herein, equation \eqref{eq:new-model} takes on the form
\begin{equation}
	\freeEnergyDensity\left(\F\right) = \min_{\rodVec} \left\{\innerFreeEnergyDensity\left(\F, \genRot\left(\rodVec\right)\right)\right\} \quad \text { subject to } \left|\rodVec\right| < 2 \pi
\end{equation}
which was solved in Mathematica using \texttt{FindMinimum} for local optimization and \texttt{NMinimize} for global optimization.
\texttt{FindMinimum} uses a series of interior point methods for constrained (local) optimization; and \texttt{NMinimize} uses Nelder-Mead methods, supplemented by differential evolution.
These choices were made for simplicity of implementation; however, we remark that the exponential representation is amenable to gradient-based methods as well~\cite{kuehnel2003minimization}.

\section{Unification of discrete polymer network models} \label{sec:unification}

In the previous section, we questioned what specifically drives the success of the $8$-chain model relative to the $3$-chain, full network, and other polymer network models.
Here we explore this in more detail through the lens of the newly proposed (local) free rotation assumption and the principle of minimum average free energy.
Towards this aim, we consider a trivial property of convex functions that we refer to as the \emph{equipartition property}.
\begin{lemma}[Equipartition property] \label{lem:equipartition}
	Given a convex function $f$ and a conserved quantity $y$ which is partitioned among $n$ variables, i.e. $y = x_1 + \dots + x_n$, a solution to \begin{equation*}
		\min_{x_1, \dots, x_n} \sum_{i=1}^n f\left(x_i\right) \; 
		\text{ subject to } \: \sum_{i=1}^n x_i = y
	\end{equation*}
	is $x_1 = \dots = x_n = \frac{y}{n}$.
\end{lemma}
\begin{proof}
	This follows directly from the definition of convexity:
	\begin{equation*}
		\sum_{i=1}^n f\left(x_i\right) = n\left(\frac{1}{n} \sum_{i=1}^n f\left(x_i\right)\right) \geq n f\left(\frac{1}{n} \sum_{i=1}^n x_i\right) = n f\left(\frac{y}{n}\right).
	\end{equation*}
\end{proof}
The equipartition property is important when considering polymer network models because of the results that follow.
A standard assumption for isotropic polymer networks is that all of the chains effectively have the same length end-to-end vectors in the stress free, undeformed configuration.
Let this length be denoted by $\Rmag$.
For continuous network models (e.g. \fref{fig:classical-network-models}.d), let $\probDens\left( \uVec \right)$ be the probability density of a chain with $\Rvec = \Rmag \uVec$ where $\uVec \in \unitSphere$ and $\unitSphere$ is the unit sphere.
\hl{Before proceeding, we also introduce the $6$-chain RVE (\fref{fig:six-chain}) because, although its elastic properties are equivalent to the $3$-chain RVE (indeed, $\chainStretch_i^2 = \left(-\Rvec_i\right)\cdot\changeCoord{\cGreen}\left(-\Rvec_i\right) / \cLen^2 = \Rvec_i\cdot\changeCoord{\cGreen}\Rvec_i / \cLen^2$), its strain gradient and stimuli-responsive properties may be different.
For example, if the polymer chains in the network have a electric (or magnetic) dipole which is linear in $\Rvec$ (e.g.~\citet{grasingerIPflexoelectricity}) then the polarization (or magnetization) of the $3$- and $6$-chain RVEs will vary.
}
\begin{proposition}[Conservation property] \label{prop:stretch-conservation}
	The sum of squares of the chain stretches, i.e. $\sum_{i=1}^{\numChains} \chainStretch_i^2$, is conserved relative to general rotations \begin{enumerate}[a)] \item for all of the $3$, $4$, $6$, and $8$-chain RVEs and \item provided the RVE has reflection symmetries about planes (passing through the origin) normal to three orthogonal directions, $\nuVec_j, \: j = 1, 2, 3$, and, $\sum_{i=1}^{\numChains} \left(\Rvec_i \cdot \nuVec_j\right)^2 = C$ for all $j$.
	For continuous networks, $\int_{\unitSphere} \df{A} \: \probDens\left(\uVec\right) \chainStretch^2\left(\uVec\right)$ is conserved provided $\int_{\unitSphere} \df{A} \: \probDens\left(\uVec\right) \left(\uVec \cdot \nuVec_j\right)^2 = C$ for all $j$.
	\end{enumerate}
\end{proposition}
While there is some overlap between a) and b), they are considered separately because a) establishes the conservation property for well known discrete polymer networks while b) is a sufficient condition for constructing new polymer network models with the same conservation property.
\begin{proof}
	\emph{a)} Let the coordinate systems for the RVEs be chosen such that
	\begin{subequations} \label{eq:Rvecs}
	\begin{equation} \label{eq:three-chain-Rvecs}
		\Rvec_i = \Rmag \euclid{i}, \quad i = 1, 2, 3,
	\end{equation}
	for the $3$-chain RVE;
	\begin{equation} \label{eq:four-chain-Rvecs}
		\begin{split}
		\Rvec_1 &= \Rmag \left(0, 0, 1\right), \quad \quad
		\Rvec_2 = \Rmag \left(0, \frac{2\sqrt{2}}{3}, -\frac{1}{3}\right) \\
		\Rvec_3 &= \Rmag \left(\sqrt{\frac{2}{3}}, -\frac{\sqrt{2}}{3}, -\frac{1}{3}\right), \quad \quad
		\Rvec_4 = \Rmag \left(-\sqrt{\frac{2}{3}}, -\frac{\sqrt{2}}{3}, -\frac{1}{3}\right)
		\end{split}
	\end{equation}
	for the $4$-chain RVE;
	\begin{equation} \label{eq:six-chain-Rvecs}
		\Rvec_i = \begin{cases}
	   \; \; \, \Rmag \euclid{i}, &\quad i = 1, 2, 3 \\
		       -\Rmag \euclid{i-3}, &\quad i = 4, 5, 6
		\end{cases}
	\end{equation}
	for the $6$-chain RVE (shown in \fref{fig:six-chain});
	and
	\begin{equation} \label{eq:eight-chain-Rvecs}
		\Rvec_i = \frac{\Rmag}{\sqrt{3}} \left(\pm 1, \pm 1, \pm 1 \right), \quad i = 1, 2, \dots, 8
	\end{equation}
	\end{subequations}
	for the $8$-chain RVE (see \fref{fig:classical-network-models}, for example).
	In each case,
	\begin{equation} \label{eq:sum-of-sq-strs}
		\sum_{i=1}^{\numChains} \chainStretch_i^2 = \frac{1}{\cLen^2} \sum_{i=1}^{\numChains} \left(\F \genRot \Rvec_i\right) \cdot \left(\F \genRot \Rvec_i\right) = \frac{1}{\cLen^2} \sum_{i=1}^{\numChains} \Rvec_i \cdot \changeCoord{\cGreen} \Rvec_i.
	\end{equation}
	where again by $\changeCoord{\cGreen}$ we mean a similarity transformation of $\cGreen$.
	Using \eqref{eq:Rvecs} and \eqref{eq:sum-of-sq-strs}, it can be shown directly that
	\begin{equation}
		\sum_{i=1}^{\numChains} \chainStretch_i^2 = \text{const} \times \trace \changeCoord{\cGreen}
	\end{equation}
	for each case, which is invariant with respect to $\genRot$, as desired\footnote{\hl{The algebra is straight forward and was verified in the Mathematica notebook, \texttt{MPS-D-22-01011.nb}, which can be found at \url{https://github.com/grasingerm/MPS-D-22-01011/}}}.
	
	\emph{b)} First consider discrete networks.
	Expand each $\Rvec$ in the orthonormal basis, $\nuVec_j, \: j = 1, 2, 3$.
	Let $\Rvec_k = \Rmag_{k1} \nuVec_1 + \Rmag_{k2} \nuVec_2 + \Rmag_{k3} \nuVec_3$ where $\Rmag_{kj} = \left(\Rvec_k \cdot \nuVec_j\right)$. Then
	\begin{equation}
		\sum_{k=1}^{\numChains} \chainStretch_k^2 = \frac{1}{\cLen^2} \sum_{k=1}^{\numChains} \Rvec_k \cdot \changeCoord{\cGreen} \Rvec_k = \frac{1}{\cLen^2} \sum_{i=1,j=1}^{3,3} \sum_{k=1}^{\numChains} \Rmag_{ki} \changeCoord{\changeCoord{\cGreenSym}}_{ij} \Rmag_{kj} = \frac{1}{\cLen^2} \sum_{i=1}^{3} \changeCoord{\changeCoord{\cGreenSym}}_{ii} \sum_k^{\numChains} \Rmag_{ki}^2.
	\end{equation}
	where $\changeCoord{\changeCoord{\cGreen}}$ is a similarity transformation of $\changeCoord{\cGreen}$ (and, consequently, is a similarity transformation of $\cGreen$) and the last step is due to the reflection symmetries.
	Clearly $\sum_{i=1}^{\numChains} \chainStretch_i^2 = \left(C / \cLen^2\right) \Tr \changeCoord{\changeCoord{\cGreen}} = \left(C / \cLen^2\right) \Tr \cGreen$, as desired.
	Similarly for continuous networks\footnote{\hl{A similar result, restricted to the case of a uniform distribution on the unit sphere, was previously shown by~\citet{kearsely1989strain} and again later by~\citet{beatty2003average}.}}, as a consequence of the assumed symmetries,
	\begin{equation}
		\int_{\unitSphere} \df{A} \: \probDens\left(\uVec\right) \chainStretch^2\left(\uVec\right) = \frac{\Rmag^2}{\cLen^2} \int_{\unitSphere} \df{A} \: \probDens\left(\uVec\right) \left(\uVec \cdot \changeCoord{\cGreen} \uVec\right) = \frac{C}{\cLen^2} \Tr \changeCoord{\cGreen}.
	\end{equation}
\end{proof}

\begin{remark}
	The $8$-chain RVE using the free rotation assumption recovers the Arruda-Boyce $8$-chain model (i.e. $8$-chain RVE using the principal frame assumption) whenever \begin{inparaenum}[1)] \item the chain free energy is a convex function of $\chainStretch^2$ and \item the chain free energy is only a function of $\chainStretch$ (i.e. does not depend on $\rdir$). \end{inparaenum}
	This follows from the equipartition property, \fref{prop:stretch-conservation}, and the formula for the classical $8$-chain constitutive model \eqref{eq:8-chain}.
\end{remark}

\begin{figure}[htb]
    \centering
    \includegraphics[width=0.5\linewidth]{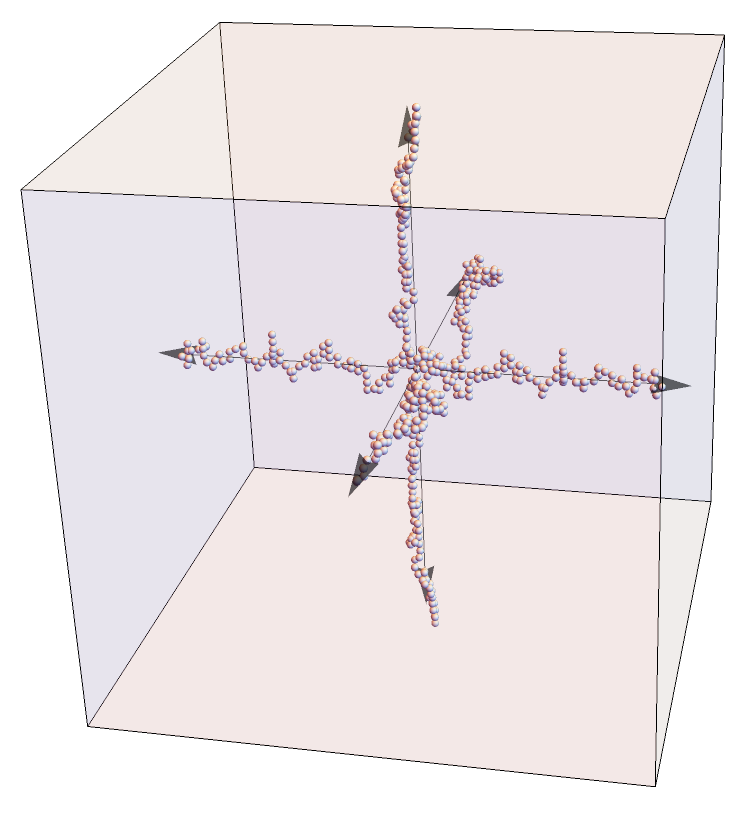}
    \caption{\hl{The $6$-chain RVE is introduced here because although its elastic properties are equivalent to the $3$-chain RVE, regardless of frame (indeed, $\chainStretch_i^2 = \left(-\Rvec_i\right)\cdot\changeCoord{\cGreen}\left(-\Rvec_i\right) / \cLen^2 = \Rvec_i\cdot\changeCoord{\cGreen}\Rvec_i / \cLen^2$), its strain gradient and stimuli-responsive properties may be different.
    For example, if the polymer chains in the network have an electric (or magnetic) dipole which is linear in $\Rvec$ (e.g.~\citet{grasingerIPflexoelectricity}) then the polarization (or magnetization) of the $3$- and $6$-chain RVEs will vary.
    }}
    \label{fig:six-chain}
\end{figure}

\begin{proposition} \label{prop:chain-free-energies-convex}
	The Gaussian chain (GC), the Kuhn and Gr\"{u}n chain (KGC), and the wormlike chain (WLC) free energies are all convex functions of $\cStrSqr = \chainStretch^2 \in \left[0, 1\right]$.
\end{proposition}
\begin{proof}
\item \paragraph{GC case} The GC is a linear function of $\chainStretch^2$.
	
\item \paragraph{WLC case} 
	Physically, it assumed that $\pLen \kB \T > 0$.
	Then
	\begin{equation}
			\derivN{\WLCFreeEnergy}{\cStrSqr}{2} = \pLen \kB \T \left(\frac{1}{4 \cStrSqr^{3/2}} - \frac{1}{4\left(1 - \sqrt{\cStrSqr}\right)^2 \cStrSqr^{3/2}} + \frac{1}{2\left(1-\sqrt{\cStrSqr}\right)^3 \cStrSqr}\right) 
			\geq \frac{\pLen \kB \T }{2\left(1-\sqrt{\cStrSqr}\right)^3 \cStrSqr} > 0.
	\end{equation}
	
\item \paragraph{KGC case} The KGC case is more difficult because an explicit expression for the inverse Langevin function, $\invLang$, does not exist.
	In terms of $\invLang$ and its derivative, $\left(\invLang\right)'$, we have that
	\begin{equation}
		\derivN{\KGFreeEnergy}{\cStrSqr}{2} = \frac{\kB \T}{4 \mLen \cStrSqr} \left(\left(\invLang\right)'\left(\sqrt{\cStrSqr}\right) - \frac{\invLang\left(\sqrt{\cStrSqr}\right)}{\sqrt{\cStrSqr}}\right).
	\end{equation}
	By assumption, $\kB \T / 4 \mLen \cStrSqr \geq 0$; so it suffices to show that the quantity $\left(\invLang\right)'\left(\sqrt{\cStrSqr}\right) - \invLang\left(\sqrt{\cStrSqr}\right) / \sqrt{\cStrSqr}$ is nonnegative.
	First, this is established numerically in~\fref{fig:KGC-convexity}.
	This can also be (approximately) supported analytically as follows.
	We use the approximation
	\begin{equation}
		\invLang\left(x\right) \approx \begin{cases}
			C_1 x + C_3 x^3 + \dots + C_{59} x^{59} & x \in [0, 0.843951) \\
			\frac{1}{1 - x} & x \in [0.843951, 1]
		\end{cases}
	\end{equation}
	which has negligible error \cite{bergstrom1999large,itskov2012taylor} (maximum absolute error $\lessapprox 10^{-8}$, maximum relative error $\leq 0.064\%$) and where, importantly, $C_j > 0$ for all $j = 1, 3, \dots, 59$~\footnote{See \cite{itskov2012taylor} for the exact coefficients.}.
	For the $\cStrSqr \in [0, 0.843951)$ case,
	\begin{equation}
		\left(\invLang\right)'\left(\sqrt{\cStrSqr}\right) - \frac{\invLang\left(\sqrt{\cStrSqr}\right)}{\sqrt{\cStrSqr}} = \sum_{k=1}^{30} \left(2k - 2\right)C_{2k-1} \cStrSqr^{k-1} \geq 0,
	\end{equation}
	because each term in the sum is nonnegative.
	Finally, for the $\cStrSqr \in [0.843951, 1]$ case,
	\begin{equation}
		\left(\invLang\right)'\left(\sqrt{\cStrSqr}\right) - \frac{\invLang\left(\sqrt{\cStrSqr}\right)}{\sqrt{\cStrSqr}} = \left(\frac{1}{\left(1 - \sqrt{\cStrSqr}\right)^2} - \frac{1}{\sqrt{\cStrSqr}\left(1 - \sqrt{\cStrSqr}\right)}\right) > 0,
	\end{equation}
	because $\sqrt{\cStrSqr} > \left(1 - \sqrt{\cStrSqr}\right)$.
	\begin{figure}[htb]
		\centering
		\includegraphics[width=\singleGraphWidth]{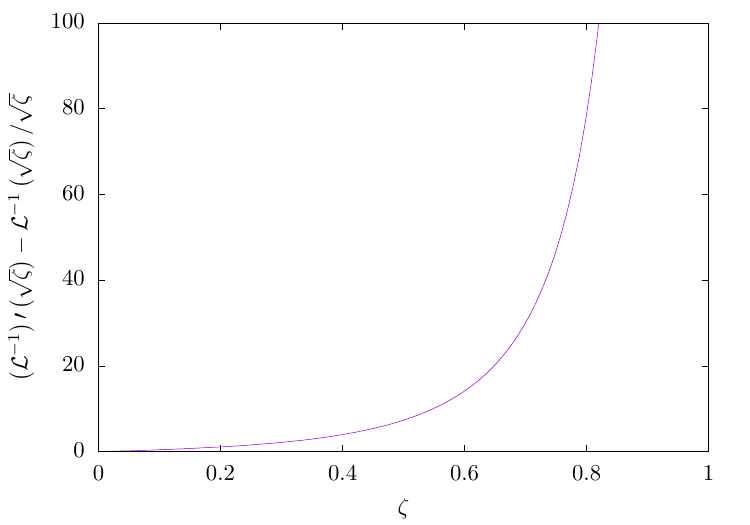}
		\caption{The quantity is nonnegative which numerically establishes the convexity of $\KGFreeEnergy$ as a function of $\cStrSqr \left(= \chainStretch^2\right)$.}
		\label{fig:KGC-convexity}
	\end{figure}
\end{proof}

\begin{proposition}[Unification of discrete polymer network models] \label{prop:unification}
	\begin{hlbreakable}
	Let \begin{equation} \label{eq:stretch-star}
		\chainStretchStar = \frac{\Rmag}{\cLen} \sqrt{\frac{\pStretch{1}^2+\pStretch{2}^2+\pStretch{3}^2}{3}}.
	\end{equation}
	\end{hlbreakable}
	Then given a polymer network which consists of chains whose free energy is \begin{inparaenum}[1)] \item a convex function of $\chainStretch^2$ and \item is only a function of $\chainStretch$ (i.e. does not depend on $\rdir$), \end{inparaenum}
	the $3$-chain, $4$-chain, $6$-chain, and $8$-chain RVEs using the free rotation assumption all produce the same constitutive model \begin{hlbreakable} because for each RVE there exists a $\genRotStar \in \SOThree$ such that
		\begin{equation} \label{eq:partition-stretch}
			\chainStretch_i = \frac{\left|\F \genRotStar \rvec_i\right|}{\cLen} = \chainStretchStar
		\end{equation}
		for all $i = 1, \dots, \numChains$ (where, recall, $\numChains$ is the number of chains in the RVE).
	\end{hlbreakable}
\end{proposition}

\begin{proof}
	\item \paragraph{\hl{$8$-chain case}}
	\hl{As discussed previously, a $\genRotStar$ which satisfies \eqref{eq:partition-stretch} exists for the $8$-chain RVE and is given by $\genRotStar = \pFrame$ where $\pFrame$ rotates the RVE into the principal frame (see \eqref{eq:8-chain})}.
	
	\item \paragraph{\hl{$3$-chain and $6$-chain cases}}
	To show that the $3$-chain and $6$-chain RVEs produce equivalent constitutive models, it suffices to show that there also exists a $\genRotStar$ with the same property; that is, a $\genRotStar$ such that all of the chains in the RVE have the same stretch given by \eqref{eq:partition-stretch}.
	The stretch for each chain in the RVE can be formulated as
	\begin{equation}
		\chainStretch_i = \frac{\left|\rvec_i\right|}{\cLen} = \frac{1}{\cLen} \sqrt{\Rvec_i \cdot \changeCoord{\cGreen} \Rvec_i}, \quad \text{ where } \changeCoord{\cGreen} = \trans{\genRot} \cGreen \genRot
	\end{equation}
	is a real proper orthogonal similarity transformation of $\cGreen$.
	There exists a real proper orthogonal similarity transformation where all of the elements on the diagonal are equal~\cite{horn1985matrix}, which for the $3$- and $6$-chain RVEs achieves equivalent stretches for all of the chains in the RVE, as desired.
	\hl{
		The existence of this similarity transformation can be understood as follows.
		Consider rotating the coordinate system about $\euclid{3}$ by $\varphi$.
		One can permute $\changeCoord{\cGreenSym}_{11}$ and $\changeCoord{\cGreenSym}_{22}$ by taking $\varphi = \pi / 2$.
		Further, this transformation of $\changeCoord{\cGreenSym}_{11}$ and $\changeCoord{\cGreenSym}_{22}$ is continuous with respect to $\varphi$ so that there exists a $\varphi$ such that $\changeCoord{\cGreenSym}_{11} = \changeCoord{\cGreenSym}_{22}$.
		Similar arguments can be made about $\changeCoord{\cGreenSym}_{11}$ and $\changeCoord{\cGreenSym}_{33}$ (by rotating about $\euclid{2}$), and about $\changeCoord{\cGreenSym}_{22}$ and $\changeCoord{\cGreenSym}_{33}$ (by rotating about $\euclid{1}$).
		Thus, there exists a $\genRotStar$ such that $\changeCoord{\cGreenSym}_{11} = \changeCoord{\cGreenSym}_{22} = \changeCoord{\cGreenSym}_{33} = 1 / 3 \Tr \cGreen$.
		Clearly, in this case, the chain stretches--for both the $3$- and $6$-chain models--satisfy \eqref{eq:partition-stretch}.
	}

	\begin{hlbreakable}
	\item \paragraph{$4$-chain case}
	To show the $4$-chain case, we decompose the rotation of interest $\genRotStar$ as $\genRotStar = \genRot' \pFrame$ where, recall, $\pFrame$ rotates the Euclidean basis to align with the principal frame such that $\changeCoord{\cGreen} = \diag \left(\pStretch{1}^2, \pStretch{2}^2, \pStretch{3}^2 \right)$.
	Then it suffices to show that $\genRot'$ exists.
	Here it is solved for explicitly.
	Let $\genRot' = \genRot\left(\eulerx \euclid{1}\right) \genRot\left(\eulery \euclid{2}\right) \genRot\left(\eulerz \euclid{3}\right)$; that is, we formulate $\genRot'$ using the Euler angle representation and $\eulerx$, $\eulery$, and $\eulerz$ are the (Euler) angles of rotation about $\euclid{1}$, $\euclid{2}$, and $\euclid{3}$, respectively.
	Then the chain with end-to-end vector $\Rvec_1 \left(= \Rmag \left(0, 0, 1\right)\right)$ (see \eqref{eq:four-chain-Rvecs}) is deformed such that $\Rvec_1 \rightarrow \F \genRot' \pFrame \Rvec_1$ and
	\begin{equation} \label{eq:stretch-41}
		\chainStretch_1 = \frac{\Rmag}{\cLen} \sqrt{\pStretch{1}^2 \sin^2 \eulery + \pStretch{2}^2 \sin^2 \eulerx \cos^2 \eulery + \pStretch{3}^2 \cos^2 \eulerx \cos^2 \eulery}.
	\end{equation}
	Here $\eulerx$ and $\eulery$ can be chosen such that $\chainStretch_1 = \chainStretchStar$.
	One such solution is $\eulerx = \pi / 4$ and $\eulery = \arccos \sqrt{2 / 3}$.
	Substituting this solution for $\eulerx$ and $\eulery$ into $\genRot'$, we see that $\Rvec_2 \left(= \Rmag \left(0, 2\sqrt{2}/3, -1/3\right)\right)$ is deformed such that
	\begin{equation} \label{eq:stretch-42}
		\chainStretch_2 = \frac{\Rmag}{9 \cLen} \sqrt{
			3\left(\pStretch{1} + 4 \pStretch{1} \sin \eulerz\right)^2 + 
				\pStretch{2}^2\left(6 \cos \eulerz + \sqrt{3}\left(1 - 2\sin \eulerz\right)\right)^2 +
				\pStretch{3}^2\left(6 \cos \eulerz - \sqrt{3}\left(1 - 2\sin \eulerz\right)\right)^2
			}.
	\end{equation}
	Now $\eulerz$ can be chosen such that $\chainStretch_2 = \chainStretchStar$.
	One such solution is $\eulerz = -\pi / 2$.
	Substituting into $\genRot'$ we see that $\chainStretch_1 = \chainStretch_2 = \chainStretch_3 = \chainStretch_4 = \chainStretchStar$~\footnote{\hl{The algebra for the $4$-chain case was verified in the Mathematica notebook, \texttt{MPS-D-22-01011.nb}, which can be found at \url{https://github.com/grasingerm/MPS-D-22-01011/}}}.
	
	A similar proof for the $4$-chain case can be found in \S 4 of~\citet{beatty2003average}. 
	There it is argued (geometrically) that the tetrahedron of the $4$-chain model can be rotated such that its chain directions lie along diagonals of the unit cube whose edges are aligned with the principal frame; such an orientation of the tetrahedron results in $\chainStretch_1 = ... = \chainStretch_4 = \chainStretchStar$.
	The mapping which rotates the $4$-chain RVE into this orientation is $\genRot' \pFrame$, as expected.
	\end{hlbreakable}
\end{proof}

\hl{As verification, we next provide some numerical results for simple deformations which agree with the above proof.}
First consider the incompressible, uniaxial deformation $\F = \diag \left(\pStretchSymbol, 1 / \sqrt{\pStretchSymbol}, 1 / \sqrt{\pStretchSymbol}\right)$.
Let the network consist of KGC chains with $\n = 100$ monomers each.
The free energy density, $\freeEnergyDensity$, as a function of $\pStretchSymbol$ for the $3$-chain, $4$-chain, and $8$-chain RVEs using the free rotation assumption are shown to agree in \fref{fig:UnifyingDiscreteNetworks_Uni}, as expected.
The deformed (and rotated) $3$-chain RVE (inset), $4$-chain RVE (right panel, top row), and $8$-chain RVE (bottom row) for $\pStretchSymbol = 0.5, 0.8,$ and $1.2$ are also shown.
Note that the $3$- and $4$-chain RVEs rotate relative to the coordinate axes but the $8$-chain RVE does not.
That is because, in this case, $\euclid{i}, i = 1, 2, 3$ makes up the principal frame.
\begin{figure}[htb]
	\centering
	\includegraphics[width=\linewidth]{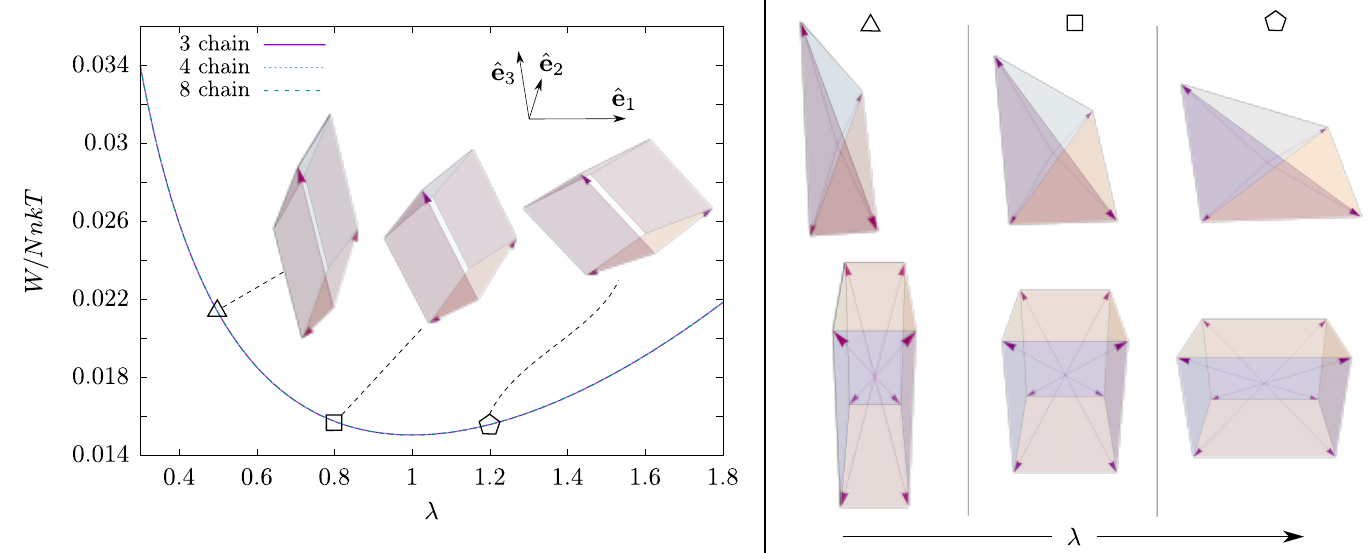}
	\caption{Free energy density as a function of uniaxial stretch, $\pStretchSymbol$, for the $3$-, $4$-, and $8$-chain RVEs with the free rotation assumption.
	Deformed (and rotated) $3$-chain RVE (inset), $4$-chain RVE (right panel, top row), and $8$-chain RVE (bottom row) for $\pStretchSymbol = 0.5, 0.8,$ and $1.2$.
	}
	\label{fig:UnifyingDiscreteNetworks_Uni}
\end{figure}

Next consider the simple shear deformation $\F = \identity + \shear \euclid{1} \otimes \euclid{3}$.
Again, the free energy densities as a function of deformation are shown to agree in \fref{fig:UnifyingDiscreteNetworks_Shear}.
\begin{figure}[htb]
	\centering
	\includegraphics[width=\linewidth]{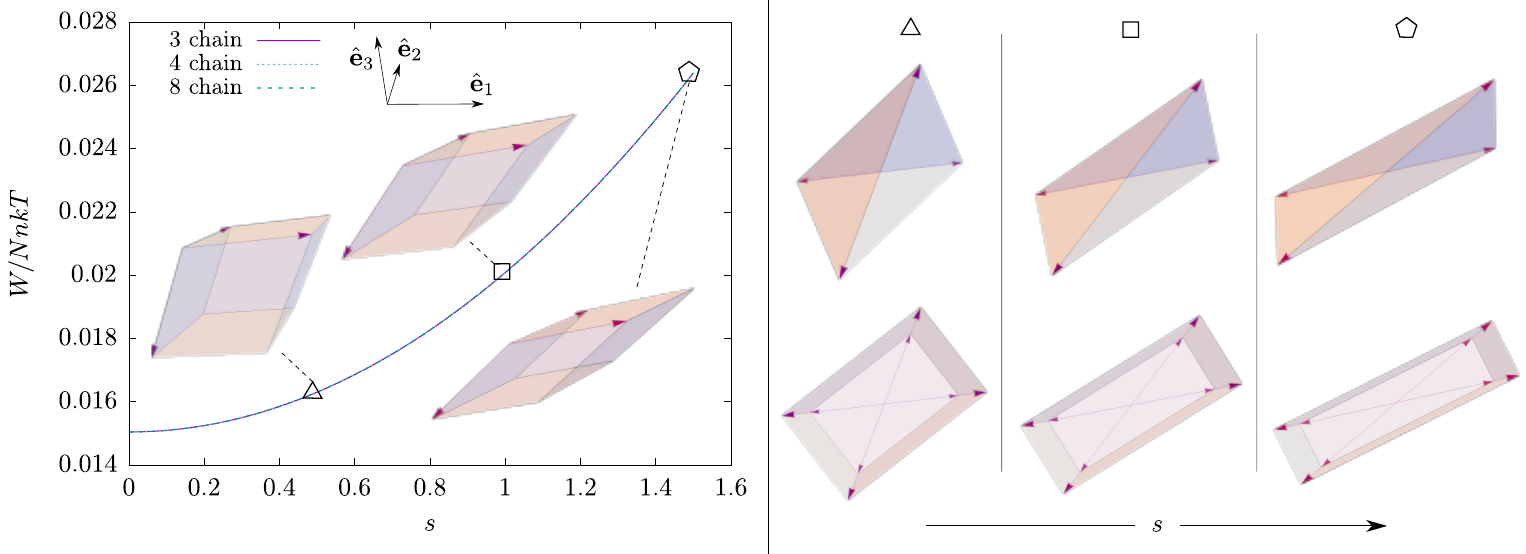}
	\caption{Free energy density as a function of shear, $\shear$, for the $3$-, $4$-, and $8$-chain RVEs with the free rotation assumption.
		Deformed (and rotated) $3$-chain RVE (inset), $4$-chain RVE (right panel, top row), and $8$-chain RVE (bottom row) for $\shear = 0.5, 1.0,$ and $1.5$.
	}
	\label{fig:UnifyingDiscreteNetworks_Shear}
\end{figure}

\begin{remark}
	\Fref{prop:unification} and the numerical results that follow are notable because they suggest that isotropic, homogeneous polymer networks which \begin{inparaenum}[1)] \item are free to locally rotate and \item consist of chains with convex free energies \end{inparaenum} have an elastic response which is insensitive to--and potentially invariant with respect to--network topology (i.e. number of chains joined at each cross-link in the network).
	Common chain free energies are convex functions of $\chainStretch^2$ (\fref{prop:chain-free-energies-convex}).
	Thus, the success of the Arruda-Boyce $8$-chain model with the principal frame assumption may be due to the fact that, coincidentally, the principal frame assumption is equivalent to the free rotation assumption. 
	The topology of the network is perhaps not the fundamental consideration; instead the choice of macro-to-micro kinematic assumption and its correspondence with free energy minimization is more fundamental.
	\hl{Concerning the topology of real networks: for chemical cross-linking, bonding typically occurs between intermediate points of a single pair of macromolecules so that it is natural to assume that the number of chains meeting at a cross-link is $4$~\cite{james1943theory,treloar1975physics}.
	As a result, it seems somewhat standard to assume degree $4$ cross-links (i.e. tetrafunctional cross-links) for many rubbers and elastomers~\cite{treloar1975physics,lei2022network,vieyres2013sulfur}.
	Other topologies, such as degree $3$ (trifunctional) and degree $8$ (octofunctional) cross-links, can also be readily obtained from modern techniques in chemical cross-linking~\cite{gu2019controlling}. 
	All of the above cross-linkers may well approximately satisfy the conditions for equipartition of squared stretch depending on their precise geometry.
	Notably, achieving higher degree cross-links (e.g. $>10$) is considered nontrivial~\cite{gu2019controlling}.
	}
\end{remark}

While the above remark may seem initially surprising given recent results that support the significance of network topology on its bulk, macroscopic properties (e.g. elastic moduli~\cite{lei2022network,alame2019relative}), we note that these results are specific to heterogeneous networks where different cross-links may be of different degree and different chains are generally of different lengths.

\subsection{Discrete vs. continuous polymer network models}

In contrast to discrete polymer network models, there are models which represent the polymer network through a continuous probability density of chains.
The full network model~\cite{wu1993improved} is the simplest such example.
It consists of a uniform distribution of chain end-to-end vectors on the sphere of radius $\Rmag$, i.e. on $\rSphere{\Rmag} = \left\{\Rvec \in \Reals^3 : \left|\Rvec\right| = \Rmag\right\}$ (see \fref{fig:classical-network-models}.d).
Here chain end-to-end vectors are often mapped under $\F$--referred to as the \emph{affine deformation assumption}~\cite{treloar1975physics}--so that the free energy density is
\begin{equation}
	\freeEnergyDensity\left(\F\right) = \frac{\chainDensity}{4\pi} \int_{\unitSphere} \df{A} \: \chainFreeEnergy\left(\frac{\Rmag}{\cLen} \sqrt{\Rdir \cdot \cGreen \Rdir} \right),
\end{equation}
where $\Rdir \in \unitSphere$ is the chain orientation.
It can be argued that, in most cases, the full network model more accurately characterizes our lack of knowledge regarding the directions of chains at any given material point of the solid body.
However, experimental data more closely agrees with the $8$-chain model with the principal frame assumption~\cite{miehe2004micro,arruda1993threee} (and, by extension, the $3$-, $4$-, $6$-, and $8$-chain RVEs with the free rotation assumption).
It is worth considering why.

The standard macro-to-micro mapping for the full network model, $\rvec = \F \Rvec$, is equivalent to the free rotation assumption, $\rvec = \F \genRot \Rvec$ because $\SOThree$ is a subgroup of the symmetries of a sphere.
Instead the principal difference between the full network and discrete networks considered previously (e.g. $3$-chain, $4$-chain, etc.) is that, for the discrete networks considered herein, there exists a $\genRot$ such that all of the chains in the RVE have the same stretch.
By the equipartition property, this has minimal average free energy for many popular chain free energy functions, and for any given $\F$.
Obviously no such $\genRot$ can exist for the full network model; and the full network model only has equal stretch throughout the distribution of chains (i.e. satisfies equipartition of squared stretch) for the special case of $\F = \alpha \identity$.
Thus, in general, the full network model predicts a larger free energy density than any of the discrete networks using the free rotation assumption.
It is possible that the predictions from discrete polymer networks match experiments more closely because they more accurately model a real cross-link which has a finite, discrete number of chains.
An interesting open question is what are the necessary conditions for a polymer network model to be able to satisfy equipartition of squared stretch (e.g. upper bound on the number of chains, symmetries, etc.)?

The microsphere approach~\cite{miehe2004micro} solves this issue for a uniform distribution of chains by modifying the macro-to-micro kinematic assumption to include a ``stretch fluctuation field''.
The fluctuation field, $\flucField : \unitSphere \mapsto \Reals$, allows for fluctuations of macroscopic line-stretch throughout the distribution of chain directions.
Here the chain stretch for a given direction, $\Rdir$, is
\begin{equation}
	\chainStretch\left(\Rdir\right) = \frac{\Rmag}{\cLen} \flucField\left(\Rdir\right) \sqrt{\Rdir \cdot \cGreen \Rdir} \quad \text{ for } \Rdir \in \unitSphere,
\end{equation}
and
\begin{subequations} \label{eq:microsphere}
\begin{align} 
	\freeEnergyDensity\left(\F \right) &= \inf_{\flucField} \avgOverR{\chainFreeEnergy\left(\frac{\Rmag}{\cLen} \flucField\left(\Rdir\right) \sqrt{\Rdir \cdot \cGreen \Rdir}\right)} \\
	\text{subject to } \avgOverR{\sqrt{\Rdir \cdot \cGreen \Rdir}}_p &= \avgOverR{\flucField\left(\Rdir\right) \sqrt{\Rdir \cdot \cGreen \Rdir}}_p,
\end{align}
where
\begin{equation}
	\avgOverR{\generic}_p = \frac{1}{4\pi} \left(\int_{\unitSphere} \df{A} \: \generic^p\right)^{1/p}.
\end{equation}
\end{subequations}
and $p > 0$ is a tunable material parameter.
\citet{miehe2004micro} show that--for chain free energies which only depend on $\chainStretch$ (and not $\rdir$)--two different local extrema of \eqref{eq:microsphere} include $\flucField = \text{const} = 1$ and $\chainStretch = \text{const} = \avgOverR{\sqrt{\Rdir \cdot \cGreen \Rdir}}_p$.
The former also requires the condition $\chainFreeEnergy' = \text{const.} \times \chainStretch^{p-1}$ and is equivalent to the full network model with $\rvec = \F \Rvec$.
The latter \begin{inparaenum}[1)] \item is optimal for convex chain free energies (due to conservation of $\chainStretch^2$), \item corresponds to constant stretch throughout the distribution of chains and, \item when $p = 2$, is equivalent to the discrete polymer network RVEs with the free rotation assumption. \end{inparaenum}

The microsphere model with the stretch fluctuation field is an elegant solution for simultaneously allowing for equipartition of squared stretch while also capturing our general lack of knowledge of the exact chain directions at a given material point.
However, the numerical integration involved is nontrivial~\cite{verron2015questioning} and solving \eqref{eq:microsphere} may be considerably more complex for chain free energies that depend on $\rdir$ as it is a constrained variational problem.
In contrast, the proposed approach of using discrete network models with the free rotation assumption (outlined in \eqref{eq:new-model}) only requires minimization over the compact, finite dimensional \hl{manifold} $\SOThree$.
In addition, the discrete polymer network models have a clear geometric interpretation provided by the RVE itself which is useful for homogenization and modeling strain gradient effects.
As a result, we will not consider the microsphere model further in this work.

\section{Volumetric torque in stimuli-responsive networks} \label{sec:stimuli-responsive}

As mentioned previously, the standard macro-to-micro kinematic assumption for discrete network models, i.e. the principal frame assumption, can make nonintuitive predictions when the polymer chains have orientational contributions to the free energy.
Orientational energies often occur in stimuli-responsive polymer networks such as networks consisting of monomers with magnetic dipoles in an externally applied magnetic field~\cite{moreno2022effects,zhao2019mechanics,bastola2021shape}, and networks consisting of monomers with electric dipoles in an externally applied electric field~\cite{grasinger2020statistical,grasinger2022statistical,cohen2018generalized}. 
Both theory and experiments suggest that these orientational energies can lead to internal torques within the network which can be used for shape morphing~\cite{moreno2022effects,zhao2019mechanics,bastola2021shape}, asymmetric actuation modes~\cite{grasingerIPtorque}, and architected polymer networks with enhanced stimuli-responsive properties~\cite{grasinger2020architected}.
An important consideration then for stimuli-responsive polymers is that macro-to-micro kinematic assumption properly represents--not only the chain stretches, but also--chain orientations for a given $\F$.

As an example stimuli-responsive polymer network, we consider dielectric elastomers.
The mechanical properties of the monomers making up each chain in the network are assumed to follow the freely jointed chain assumptions. 
Additionally, the monomers are assumed to be dielectric; that is, bound charges on a monomer can be separated to form an electric dipole, $\dipole$.
Let $\nvec$ denote the orientation of a monomer.
Given that the dipole depends on the magnitude of the electric field and the orientation of the monomer, $\nvec$, relative to the direction of the electric field, we use a simple anisotropic form \cite{stockmayer1967dielectric,cohen2016electroelasticity}:
\begin{equation}
		\dipole\left(\nvec, \eFieldVec\right) = \vperm \dipoleSusceptibility \eFieldVec = \vperm \left[\sustensexpr\right] \eFieldVec 
\end{equation}
where $\dipoleSusceptibility$ is the dipole susceptibility tensor, $\susPara$ and $\susPerp$ are the dipole susceptibility along $\nvec$ and the susceptibility in plane orthogonal to $\nvec$, respectively, $\eFieldVec$ is the local electric field, and $\vperm$ is the vacuum permittivity.
We refer to monomers with $\susPara > \susPerp$ as field-aligning (FA) and monomers with $\susPerp > \susPara$ as field-disaligning (FD). 

The energy of a single monomer has two contributions: the energy associated with separating charges and the electric potential of a dipole in an electric field~\cite{grasinger2020statistical,grasinger2022statistical}~\footnote{Where $\dipoleSusceptibility^{-1}$ denotes the generalized inverse of $\dipoleSusceptibility$.}:
\begin{equation}  \label{eq:monomer-energy}
	\um = \frac{1}{2\vperm} \dipole \cdot \dipoleSusceptibility^{-1} \dipole - \dipole \cdot \eFieldVec = \frac{\vperm \dsus}{2} \left(\eFieldVec \cdot \nvec\right)^2 - \frac{\vperm \susPerp}{2} \eFieldMag^2,
\end{equation}
where $\dsus = \susPerp - \susPara$.
We assume the monomer energy is such that $\um \sim \kB \T$ and that the energy of dipole-dipole interactions are much less than the thermal energy, i.e. $\vperm \left[\eFieldMag \times \max \left(\left\{\susPara, \susPerp\right\}\right) \right]^2 / \mLen^3 \ll \kB \T$. 
The latter is therefore neglected.

Given the electrostatic energy of a single monomer, the free energy (assuming constant $\T$, $\eFieldVec$, and $\rvec$) is derived using statistical mechanics~\cite{grasinger2020statistical}:
\begin{subequations} \label{eq:DEFreeEnergy}
\begin{equation}
	\DEFreeEnergy\left(\rvec, \eFieldVec\right) = \KGFreeEnergy\left(\chainStretch\right) + \elecFreeEnergy\left(\chainStretch, \eFieldVec\right) + \orientFreeEnergy\left(\chainStretch, \eFieldVec, \rdir\right)
\end{equation}
where $\KGFreeEnergy$ is the purely mechanical contribution to the free energy (see \eqref{eq:KGC}), and
\begin{align}
	\elecFreeEnergy &= \n \kB \T \left\{\frac{\chainStretch \unodim}{\zmultzero} - \uOnodim + \left(1 - \chainStretch^2 \right) \left[-\frac{\unodim}{3} + \ln \left(\frac{2 \sqrt{\unodim}}{\sqrt{\pi} \erfw}\right)\right]\right\}\\
	\orientFreeEnergy &= \n \kB \T \left\{\unodim \left(1 - \frac{3 \chainStretch}{\zmultzero}\right) \left(\eDir \cdot \rdir\right)^2 \right\}
\end{align}
\end{subequations}
are the electric-stretch and electric-orientation contributions, respectfully; and where $\unodim = \eFieldMag^2 \dsus / 2 \kB T$ and $\uOnodim = \eFieldMag^2 \susPerp / 2 \kB T$.
Now let $\polar$ be the angle between $\rdir$ and $\eDir$.
The electrostatic torque on the dielectric chain can be obtained by $-\partialInline{\DEFreeEnergy}{\polar}$.
For chains consisting of FA monomers (i.e. $\unodim < 0$), the torque is driving $\rdir$ toward $\eDir$ when $\polar < \pi / 2$ and $-\eDir$ when $\polar > \pi / 2$--its directions of minimal free energy.
For chains consisting of FD monomers (i.e. $\unodim > 0$), the torque is driving $\rdir$ toward the plane orthogonal to $\eDir$~\cite{grasinger2020statistical,grasinger2022statistical}.
Note that, in the absence of an electric field (i.e. $\unodim = 0$), $\DEFreeEnergy$ reduces to the Kuhn and Gr\"{u}n approximation for a freely jointed chain.
Given $\DEFreeEnergy$, the free energy density for the model developed in this work takes the form
\begin{subequations} \label{eq:new-DE-model}
	\begin{align}
		\label{eq:new-DE-model-min-rot}
		\freeEnergyDensity\left(\F, \eFieldVec\right) &= \min_{\rodVec} \innerFreeEnergyDensity\left(\F, \eFieldVec, \rodVec\right) \quad \text{subject to} \quad \left|\rodVec\right| < 2 \pi,\\
		\label{eq:new-DE-model-inner}
		\innerFreeEnergyDensity\left(\F, \eFieldVec, \rodVec\right) &= \chainDensity \avgOverR{\DEFreeEnergy\left(\F \genRot\left(\rodVec\right) \Rvec, \eFieldVec\right)}.
	\end{align}
\end{subequations}

As an illustration, consider a thin DE film (thickness $h$) with compliant electrodes on its top and bottom surfaces where the film is constrained such that it can only undergo homogeneous simple shear deformation.
We also assume that the electric field is fixed and across the thickness of the film--which, neglecting fringe effects, would be realized by applying a voltage difference across the electrodes, $\Delta \phi$.
A schematic is shown in \fref{fig:shear}.
\begin{figure}
    \centering
	\includegraphics[width=0.8\linewidth]{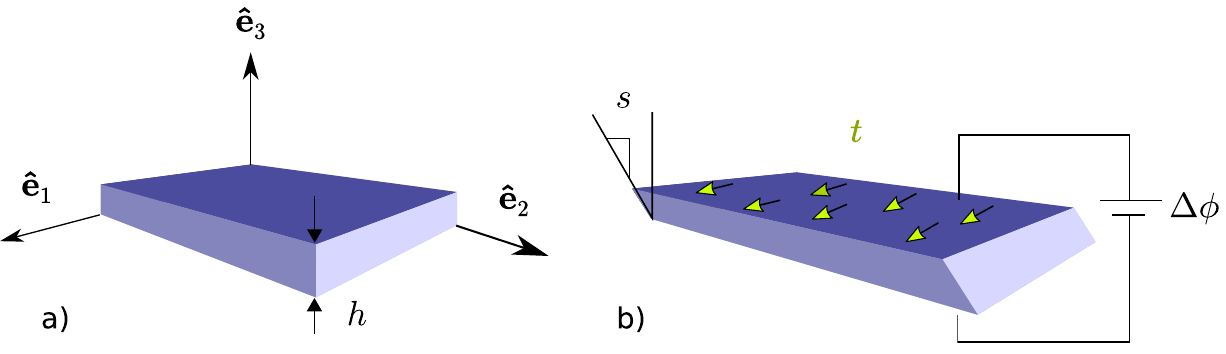}
	\caption{Shear-mode actuation of a dielectric elastomer actuator constrained to only undergo shear deformation.}  
	\label{fig:shear}
\end{figure}

Recall for simple shear, the deformation is of the form $\F = \identity + \shearDef \euclid{1} \otimes \euclid{3}$.
Let a traction, $\tvec = \tmag \euclid{1}$, be applied to the top surface of the film.
The energy density of the electric field is neglected because it does not do work on the DE.
By homogeneity, the free energy is $\vol \times \left(\freeEnergyDensity - \tmag \shearDef\right)$ where $\vol$ is the volume of the body.

Let $\tnodim \coloneqq \tmag / \chainDensity \kB \T$ be the nondimensional applied traction and $\n = 25$ monomers.
Now, following~\citet{grasingerIPtorque}, consider the following \hl{in-silico} experiment:
\begin{inparaenum}[1)]
	\item the traction is applied and is held constant;
	\item the traction causes initial shear deformation, $\shearDef_0$;
	\item a voltage difference is applied across the electrodes such that there is electric field, $\eFieldVec = \left(-\Delta \phi / h\right) \euclid{3}$, in the DE;
	\item for each $\eFieldVec$, a shear strain $\shearDef$ is observed.
\end{inparaenum}
The deformation as function of $\unodim$ predicted by the $3$-, $4$-, \hl{$6$-} and $8$-chain RVEs with the free rotation assumption (depicted by magenta `$\triangle$', green `$\square$', \hl{red `$\pentagon$',} and blue `$*$', respectively), and the full network model with the affine deformation assumption (yellow `$\ocircle$'), is shown in \fref{fig:shear-def}.a.
For the solid lines $\tnodim = 1$, and for the dashed lines $\tnodim = 2$.
There is near exact agreement between the discrete network models, and the full network model also agrees well.
One can see that, for networks with field disaligning (FD) monomers (i.e. $\unodim > 0$), the shear deformation increases with increasing electric field (i.e. $\left|\unodim\right| \uparrow$).
In contrast, for networks with field aligning (FA) monomers (i.e. $\unodim < 0$), the shear deformation decreases with increasing electric field.
Physically, this can be understood as follows.
Intuitively, the shear displacement of the top surface relative to the bottom surface serves to ``flatten'' chains into the plane orthogonal to the direction of the electric field, $\euclid{3}$.
Thus, the shear displacement is electrostatic-energetically favorable (i.e. decreases $\elecFreeEnergy$) for networks with FD monomers and unfavorable for networks with FA monomers~\cite{grasingerIPtorque}.
\begin{figure*}
    \centering
	\includegraphics[width=\linewidth]{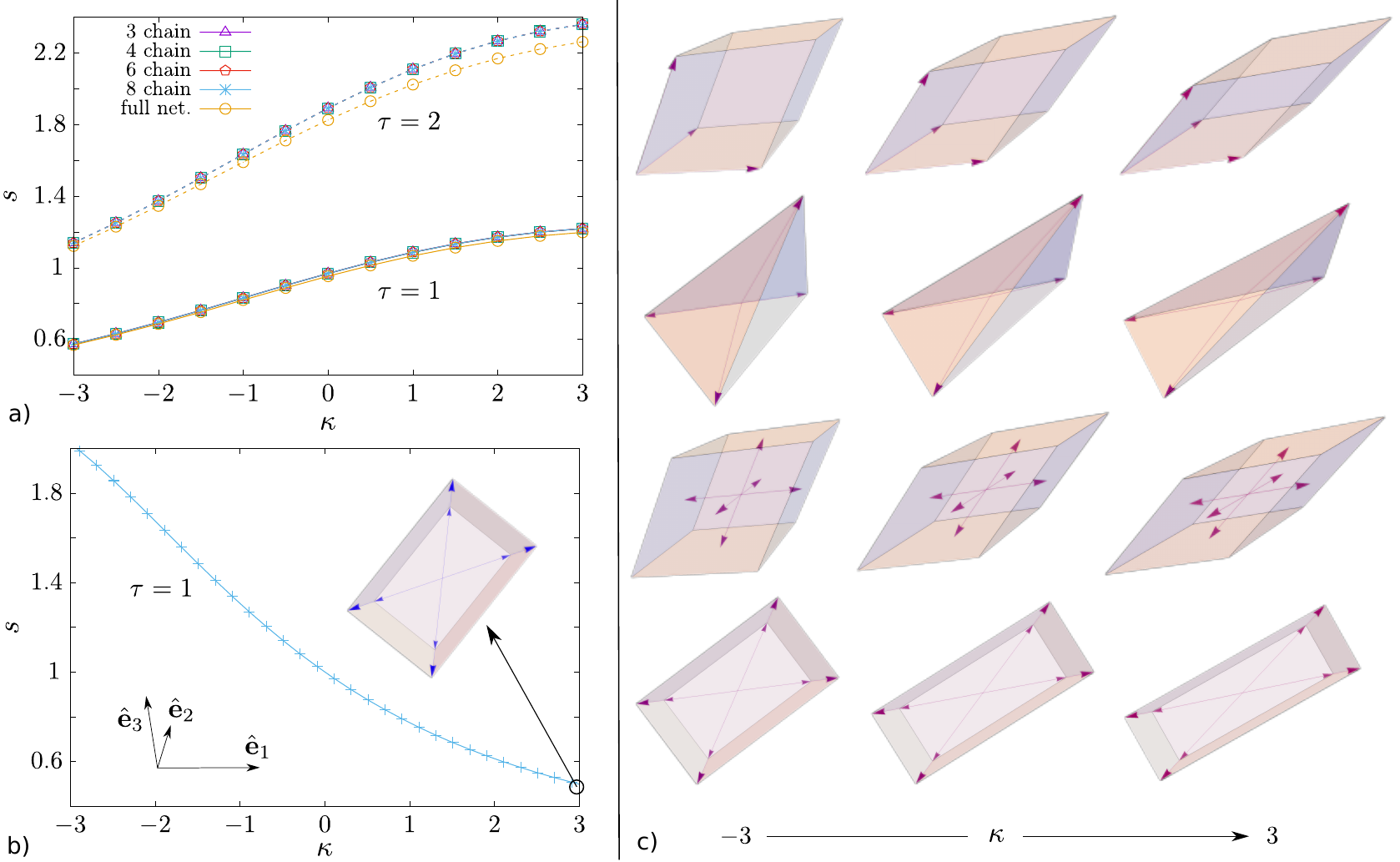}
	\caption{\capTitle{Shear-mode actuation of dielectric elastomers.}
		a) Shear deformation, $\shearDef$, as a function of $\unodim$ predicted by the $3$-, $4$-, \hl{$6$-}, and $8$-chain RVEs with the free rotation assumption and the full network model with the affine deformation assumption.
		For networks with FD monomers (i.e. $\unodim > 0$), the shear deformation increases with increasing electric field; and for networks with FA monomers (i.e. $\unodim < 0$), the shear deformation decreases with increasing electric field.
		This is because the shear displacement serves to ``flatten'' chains into the plane orthogonal to the electric field.
		b) Shear deformation as predicted by the classical $8$-chain model \hl{with the principal frame assumption} is diametrically opposed to the intuitive results of a).
		c) Deformed RVEs for the $3$- (top row), $4$- (top-middle row), \hl{$6$- (bottom-middle row),} and $8$-chain (bottom row) RVEs for $\unodim = -3, 0,$ and $3$ (left column, middle column, and right column, respectively).
	}  
	\label{fig:shear-def}
\end{figure*}

The disagreement of the full network model (with respect to the discrete networks) is small but significant.
The full network model underestimates the deformation, relative to the discrete networks, because the continuous uniform density of chains cannot reduce its free energy by a local rotation.
The free energy of the full network (for the same $\chainDensity$) is greater so that, in balancing with the work of the applied traction, less deformation occurs and the network is apparently stiffer.

Next consider the deformation predicted by the classical Arruda-Boyce $8$-chain model; that is, the $8$-chain RVE using the principal frame assumption--as shown in \fref{fig:shear-def}.b for $\tnodim = 1$.
This predicted deformation is diametrically opposed to the predictions of the $3$-, $4$-, \hl{$6$-}, and $8$-chain RVEs with the free rotation assumption and is counter-intuitive.
Not only is the free rotation assumption more consistent with thermodynamic principles than the principal frame assumption, but this also suggests that the free rotation assumption more properly represents chain orientations within the network for a given $\F$.
\Fref{fig:shear-def}.c shows the deformed RVEs for the $3$- (top row), $4$- (top-middle row), \hl{$6$- (bottom-middle row),} and $8$-chain (bottom row) RVEs for $\unodim = -3, 0,$ and $3$ (left column, middle column, and right column, respectively).

A quantitative measure of the RVE rotations as a function of $\unodim$ is given in \fref{fig:shear-rotations}.
Let $\rodVec_0$ be the Rodrigues' vector for the RVE rotation at zero applied electric field, $\unodim = 0$, and (nonzero) applied traction $\tnodim$.
Then \fref{fig:shear-rotations} shows $\left|\rodVec - \rodVec_0\right|$ as a function of $\unodim$ for the $3$-, $4$-, \hl{$6$-}, and $8$-chain RVEs (depicted by magenta `$\triangle$', green `$\square$', \hl{red `$\pentagon$',} and blue `$*$', respectively).
For the dash-dot lines $\tnodim = 0.5$, and for the solid lines $\tnodim = 1$.
It is notable that whether the $4$-chain or $8$-chain RVE experiences more rotation as a function of $\unodim$ depends on $\tnodim$; and the rotation of the $3$- and \hl{$6$}-chain RVEs (relative to $\rodVec_0$) are \begin{inparaenum}[1)] \item \hl{equal, and} \item not monotonic in $\left|\unodim\right|$. \end{inparaenum}
\hl{The rotations for the $3$- and $6$-chain RVEs agree exactly in this case because the chain free energy has the symmetry $\DEFreeEnergy\left(\rvec\right) = \DEFreeEnergy\left(-\rvec\right)$.}
\begin{figure}
    \centering
	\includegraphics[width=\singleGraphWidth]{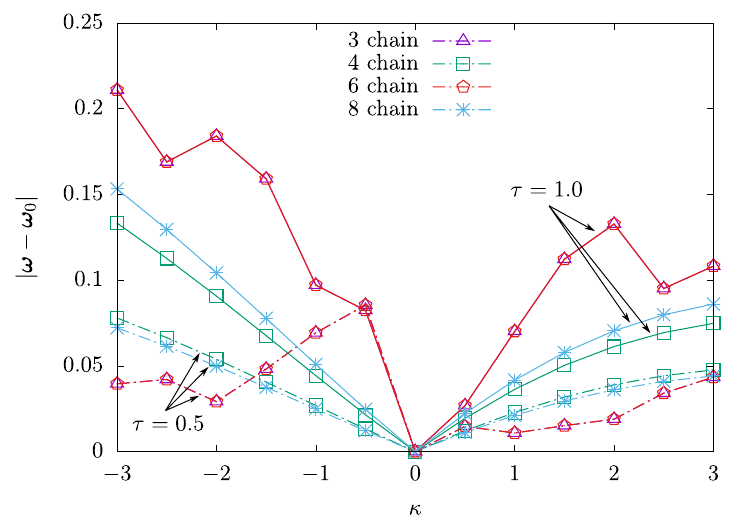}
	\caption{\capTitle{Local RVE rotation of shear-mode actuated dielectric elastomers.}
		$\left|\rodVec - \rodVec_0\right|$ as a function of $\unodim$ for the $3$-, $4$-, \hl{$6$-,} and $8$-chain RVEs where $\rodVec_0$ is the Rodrigues' vector for the RVE rotation at zero applied electric field.
		For the dash-dot lines $\tnodim = 0.5$, and for the solid lines $\tnodim = 1$.
		Whether the $4$-chain or $8$-chain RVE experiences more rotation as a function of $\unodim$ depends on $\tnodim$; and the rotation of the $3$- and \hl{$6$}-chain RVEs (relative to $\rodVec_0$) are 1) \hl{equal, and} 2) not monotonic in $\left|\unodim\right|$.
	}  
	\label{fig:shear-rotations}
\end{figure}

\section{Phases and multistability in networks with nonconvex chain free energies: biopolymers and semi-crystalline polymers} \label{sec:nonconvex}

It was shown in \fref{sec:unification} that, for chain free energies which are convex functions of $\chainStretch^2$ (e.g. GC, WLC, KGC), \begin{inparaenum}[1)] \item \label{item:unification} the free rotation assumption unifies discrete polymer network RVEs and \item \label{item:equipartition} rotations after which $\F$ stretches each of the chains equally (i.e. equipartition of squared stretch is satisfied) minimize the average free energy. \end{inparaenum}
A natural question then is: ``are there (physically meaningful) nonconvex $\chainFreeEnergy$ for which either of the above results do not hold?''

\newcommand{\bioC}{\epsilon}
\paragraph*{Biopolymers} Many biopolymers such as fibrin and collagen fibers are known to buckle in compression~\cite{sun2020fibrous,broedersz2014modeling}; and, as a result, have a nonconvex free energy.
Recently, fiber buckling has been proposed as a mechanism for certain densification phase transitions that occur in biopolymer networks resulting in cell remodelling, tethering between cells, and the compression response of blood clots~\cite{grekas2021cells,sun2020fibrous}.
Here we follow recent work~\cite{grekas2021cells,sun2020fibrous} and assume that the chain free energy is a quartic function of $\chainStretch$.
Let $\relChainStretch = \rmag / \Rmag$ denote the relative chain stretch.
We assume the form
\begin{equation}
	\bioFreeEnergy\left(\relChainStretch\right) = \frac{\alpha}{\bioC^4} \kB \T \left(\left(\relChainStretch - 1 + \bioC\right)^2 - \bioC^2\right)^2.
\end{equation}
where $0 < \bioC < 1$.
This $\bioFreeEnergy$ has local minima at $\relChainStretch = 1$ and the compressed length $\relChainStretch = 1 - \bioC$ such that it captures fiber buckling.
Here we choose $\bioC = 0.05$.

Next, consider the deformation gradient $\F = \diag\left(\pStretchSymbol, 1, 1\right)$ where $0 < \pStretchSymbol < 1$ characterizes the compression of a biopolymer gel~\cite{sun2020fibrous} with vanishing bulk modulus.
\Fref{fig:Bio} shows the free energy density as a function of $\pStretchSymbol$ for the $3$-, $4$-, $6$-, and $8$-chain RVEs with the free rotation assumption, and for the classical $8$-chain model with the principal frame assumption (PFA) (i.e. all chains have the same stretch).
Local minima correspond with different phases.
As can be seen in \fref{fig:Bio}, the number of phases depends on the macro-to-micro kinematic assumption.
The free rotation networks all have $4$ phases while the $8$-chain with PFA has $2$.
Because the RVEs with the free rotation assumption do not agree with the standard $8$-chain model, we can conclude that there are indeed interesting nonconvex chain free energies for which equipartition of squared stretch is not satisfied.
Specifically, not one of the RVEs considered satisfies the equipartition property within approximately the domain of $\pStretchSymbol \in \left(\approx 0.75, \approx 0.95\right)$.
\hl{Instead, the network obtains a lower energy by relaxing to a state in which some biopolymers are buckled while others are less compressed.}
Recent results suggest that realistic biopolymer networks may indeed not have equal stretch for chains in various directions (i.e. not agree with equipartition)~\cite{song2022hyperelastic}.
In contrast, all of the network models agree when $\pStretchSymbol < 0.7$ or $\pStretchSymbol > 0.96$ which suggests the equipartition solution is the optimal one for these ranges of $\pStretchSymbol$.
While all of the free rotation polymer networks have the same number of phases, the location and depth of the phases depends on the network topology.
The free energy response of $4$- and $8$-chain RVEs agree and the $3$- and $6$-chain models agree; however, the $4$- and $8$-chain differ from the $3$- and $6$-chain.
Thus, there are instances for which, not only is the macro-to-micro deformation relationship significant, but the network topology is significant as well.
\hl{That the $3$- and $6$-chain RVEs agree is a straight forward consequence of the fact that \begin{inparaenum}[1)] \item $\bioFreeEnergy$ depends only on chain stretch and \item $\chainStretch_i = \sqrt{\left(-\Rvec_i\right)\cdot\changeCoord{\cGreen}\left(-\Rvec_i\right)} / \cLen = \sqrt{\Rvec_i\cdot\changeCoord{\cGreen}\Rvec_i} / \cLen$. \end{inparaenum}
Similarly, the $4$- and $8$-chain RVEs agree because the $4$-chain tetrahedron can be rotated (relative to the $8$-chain RVE) such that its chains lie along the diagonals--and, hence, chains--of the $8$-chain RVE (see the proof of \fref{prop:unification} or~\citet{beatty2003average}).
We can conclude then that \emph{when the chain free energy is only a function of stretch, the RVEs can be separated into two equivalence classes: one consisting of the $3$- and $6$-chain RVEs, and the other consisting of the $4$- and $8$-chain RVEs}.
The two equivalence classes merge when the chain free energy is a convex function of $\chainStretch^2$.
}
\begin{figure}
    \centering
	\includegraphics[width=\singleGraphWidth]{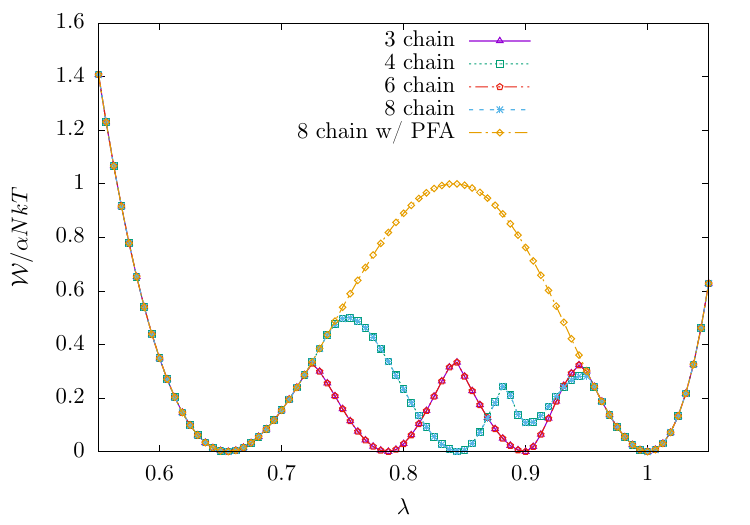}
	\caption{
		\capTitle{Densification phases of a biopolymer gel.}
		Free energy response of a biopolymer gel (with fiber buckling) to a $1$ dimensional, vanishing bulk modulus, deformation as predicted by $3$-, $4$-, $6$-, and $8$-chain RVEs with the free rotation assumption and by the classical $8$-chain model with the principal frame assumption \hl{(PFA)}.
		The number of predicted phases is shown to depend on the macro-to-micro kinematic assumption; and the depth and location of the free energy wells depends on the network topology.
	}  
	\label{fig:Bio}
\end{figure}

\hl{
It should be noted that the above model is minimal and likely missing key some elements needed to characterize more general deformation modes in biopolymer networks.
For instance, it is well known that certain biopolymer networks have deformation regimes dominated by fiber bending rather than affine stretching or compression~\cite{broedersz2014modeling,abhilash2014remodeling,pritchard2014mechanics}.
Without considering nonaffine affects such as the bending stiffnesses of the fibers or transmission of bending moments at cross-links~\cite{pritchard2014mechanics}, nor the nonvanishing bulk modulus of the network, the resulting model can produce nonphysical results such as zero resistance to infinitesimal simple shear~\cite{song2022hyperelastic}; and, in fact, the $3$- and $6$-chain RVEs become mechanisms to pure shear~\cite{friesecke2002validity}.
Complicating matters further, connectivity in biopolymer networks is sometimes subcritical and can be transient~\cite{broedersz2014modeling,pritchard2014mechanics}.
We also emphasize that fiber buckling is not the only mechanism relevant for the emergence of macroscale mechanical instabilities in biopolymer networks.
Instead the above model serves as a minimal example for instances where the equipartition of stretch is not the optimal deformation of the network and associated consequences for certain densification phases in biopolymer networks.
}

\paragraph*{Semi-crystalline polymers} As a final example, we probe the generality of the above result by considering a chain free energy which, in addition to fiber buckling, also has a stretch-induced crystallization phase.
\begin{equation}
	\chainFreeEnergy_{\text{semi-crystal}}\left(\relChainStretch\right) = \frac{\alpha}{\bioC^6} \kB \T \left(\left(\relChainStretch - 1\right)^2 - \bioC^2\right)^2 \left(\relChainStretch - 1\right)^2.
\end{equation}
where, again, $\bioC = 0.05$.
For this case, the chain has three minima: $\relChainStretch = 1, 1 \pm \bioC$.
In contrast to the biopolymer gel with vanishing bulk modulus, we consider an incompressible deformation of the form $\F = \diag\left(\pStretchSymbol, 1 / \sqrt{\pStretchSymbol}, 1 / \sqrt{\pStretchSymbol}\right)$.
Again, in \fref{fig:Bio3}, we see similar phenomena such as \begin{inparaenum}[1)] \item the free rotation networks predict more phases than the \hl{$8$-chain RVE with the PFA}, \item the free rotation networks satisfy equipartition (i.e. are equivalent to the \hl{$8$-chain with PFA}) for certain intervals of deformation, $\pStretchSymbol \in \left(0, \approx 0.71\right) \cup \left(\approx 0.89, \approx 1.15\right) \cup \left(\approx 1.33, \infty\right)$, and \item the $4$- and $8$-chain RVEs agree well, and $3$- and $6$-chain RVEs agree well \hl{(as expected)}, but the $4/8$ behaves differently than the $3/6$ outside of the equipartition intervals. \end{inparaenum}
\begin{figure}
    \centering
	\includegraphics[width=\singleGraphWidth]{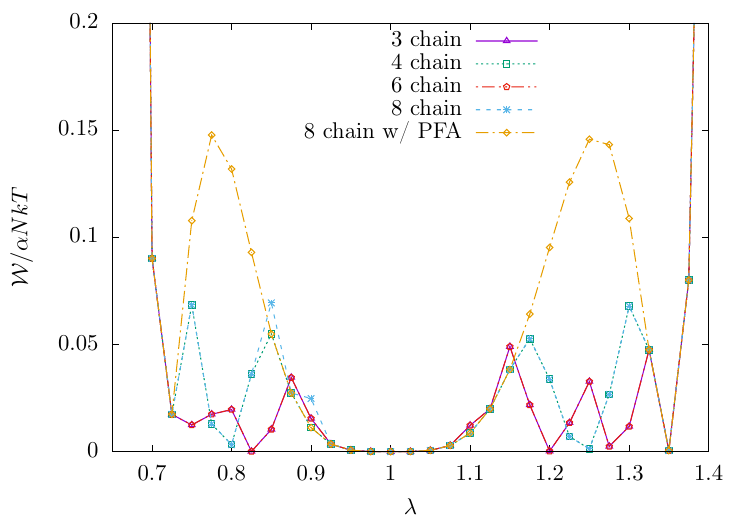}
	\caption{
		\capTitle{Phases of a semi-crystalline polymer network.}
		Free energy response of a semi-crystalline polymer network (with fiber buckling and strain-induced crystallization) to an incompressible, uniaxial deformation as predicted by $3$-, $4$-, $6$-, and $8$-chain RVEs with the free rotation assumption and by the classical $8$-chain model with the principal frame assumption \hl{(PFA)}.
		The number of predicted phases, and the depth and location of the free energy wells, depend on both the macro-to-micro kinematic assumption and on the network topology.
	}  
	\label{fig:Bio3}
\end{figure}

\section{Conclusion} \label{sec:conclusion}

Recent work has highlighted the importance of polymer chain orientational energies and internal torques within the network as they have the potential to be used for shape morphing~\cite{moreno2022effects,zhao2019mechanics}, asymmetric actuation modes~\cite{grasingerIPtorque}, and architected polymer networks with enhanced stimuli-responsive properties~\cite{grasinger2020architected}.
However, two of the most important and successful polymer network models, the \hl{$8$-chain with the principal frame assumption (i.e. Arruda-Boyce)} and the \hl{affine} full network models, make contradictory predictions for stimuli-responsive polymer networks.
This motivates the need for new approaches to polymer network modelling.

In this work, we developed a new macro-to-micro kinematic assumption for relating macroscopic deformation, $\F$, to the deformation of chains within the polymer network.
The newly developed kinematic assumption, the free rotation assumption, posits that the network is free to locally rotate (at material point $\xvec$) in order to most efficiently accommodate some $\F$.
Importantly, the newly developed model assumptions do not require any additional fitting parameters.
We showed that the free rotation assumption unifies discrete polymer network models (specifically, the $3$-, $4$-, $6$-, and $8$-chain RVEs) and recovers the successful \hl{$8$-chain model with the principal frame assumption} when the network consists of chains which do not depend on $\rdir$ and are convex functions of $\chainStretch^2$.

Next we modelled the constitutive response of dielectric elastomers using the newly developed free rotation polymer network models and compared to classical polymer network models.
It was found that the free rotation assumption more closely agrees with intuition for an asymmetric, shear electromechanical actuation mode.
Lastly, we explored the implications of the free rotation assumption for phases and multistability in networks with nonconvex chain free energies; this included models for fiber buckling in biopolymer networks and strain induced crystallization of semi-crystalline networks.
It was found, in each case, that the number of phases depended on the macro-to-micro kinematic assumption (i.e. free rotation vs. principal frame).
Also, for certain ranges of deformation, the network topology was significant as there was a difference in the free energy response between the $3$ and $6$-chain RVEs, and the $4$- and $8$-chain RVEs.
For the stimuli-responsive, biopolymer, and semi-crystalline polymer networks considered in this work, the difference between the newly developed discrete polymer network models and the \hl{classical $8$-chain model with the principal frame assumption} were more than quantitative, they were also \emph{qualitative}.
This suggests it may be possible to validate, via careful experimentation, which macro-to-micro kinematic assumptions are more or less physically accurate for a given polymer network; in other words, the inherent physical phenomena may be isolated from the fitting of model parameters.
While in some cases it may be difficult to directly discern or differentiate stable phases of a multistable polymer network, it may be possible to leverage new statistical techniques inspired by work on early warning signals to aid in the measurement of phases~\cite{matsui2022visualizing}.

Interesting topics for potential future work include the implications of the free rotation assumption for liquid crystal elastomers, the fracture of polymer networks, and strain gradient effects like flexoelectricity or strain gradient elasticity.
Although excluded volume effects and chain entanglement were not explicitly considered in this work, we do remark that the tube model (e.g.~\cite{miehe2004micro}), and other, more physics-based approaches for excluded volume (e.g.~\cite{khandagale2022statistical}), may be naturally incorporated into the developed network modelling approach.

\section*{Software availability}
The code(s) used for analysis and generation of data for this work is available at \url{https://github.com/grasingerm/MPS-D-22-01011}.

\section*{Acknowledgments}
 The author acknowledges insightful discussions with Kaushik Dayal and Benjamin Grossman, and the support of the Air Force Research Laboratory.

\appendix

\begin{hlbreakable}
\section{Inhomogeneous deformations and compatibility} \label{app:compatibility}

In this appendix, the issue of compatibility in a polymer network undergoing inhomogeneous local rotations is explored.
Consider two $8$-chain RVEs side-by-side along the $\euclid{1}$ direction, as shown in \fref{fig:compatibility} (side view).
The positions of the cross-links at the center or each RVE are $\xvec$ (left RVE) and $\xvec + 2 \Rmag / \sqrt{3} \euclid{1}$ (right RVE).
Four of the chains in the left RVE share a cross-link with a corresponding chain in the right RVE.
Let $\Rvec_1$ and $\Rvec_2$ be a pair of reference chain end-to-end vectors from the left and right RVEs, respectively, which share a common cross-link.
Given a general deformation field $\F\left(\xvec\right)$, let $\genRotStar = \arg \min_{\genRot \in \SOThree} \innerFreeEnergyDensity\left(\F, \genRot\right)$ be the corresponding local rotation field.
A standard assumption is that $\F$ varies ``gradually'' relative to the size of the RVE~\cite{grasingerIPflexoelectricity}.
Otherwise the continuum constitutive response to $\takeGrad{\F}$ produced by the network model can be nonphysical, and, in general, a continuum-scale description is likely infeasible~\cite{grasingerIPflexoelectricity}.
A natural question then is under what additional conditions a gradual change in $\F$ also implies a gradual change in $\genRotStar$; that is, ``when does $\lVert \takeGrad{\F} \rVert \Rmag \ll 1 \implies \lVert \takeGrad{\genRotStar} \rVert \Rmag \ll 1$?''
For the special cases of the $4$-chain and $8$-chain RVEs with chain free energies which are convex functions of $\chainStretch^2$ and do not depend on chain orientation, recall that $\genRotStar = \genRot' \pFrame$ (where $\genRot'$ is a constant) and $\genRotStar = \pFrame$, respectively.
The principal directions of stretch are continuous functions of $\F$~\cite{horn1985matrix} so that $\lVert \takeGrad{\genRotStar} \rVert \Rmag$ can be made small provided that $\lVert \takeGrad{\F} \rVert \Rmag$ is small enough.
How generally this holds remains an open and difficult question, however, this illustrates that there are indeed important physical cases for which $\lVert \takeGrad{\F} \rVert \Rmag \ll 1 \implies \lVert \takeGrad{\genRotStar} \rVert \Rmag \ll 1$.
\begin{figure}
	\centering
	\includegraphics[width=0.9\linewidth]{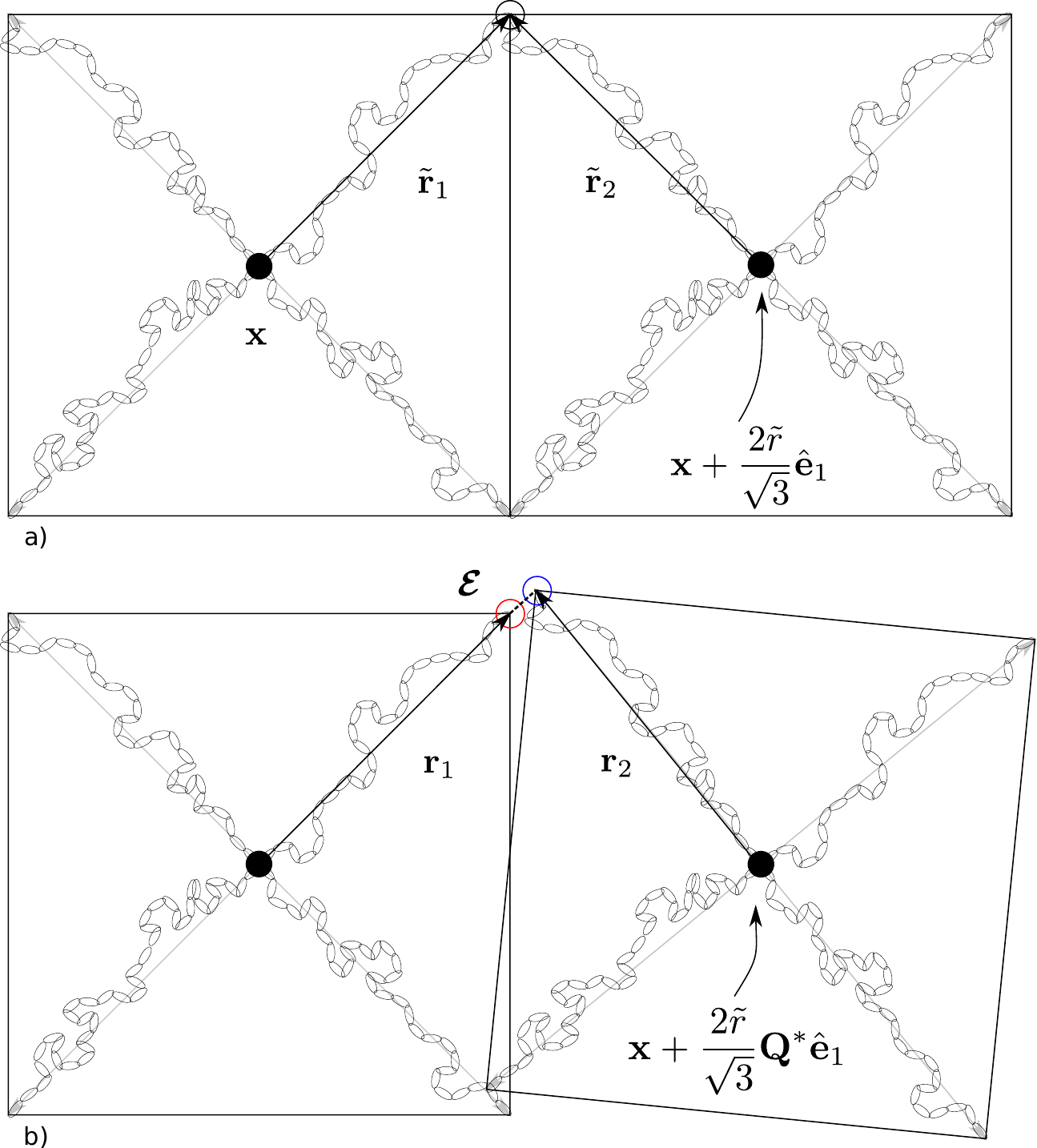}
	\caption{
		\hl{
			\capTitle{Network incompatibility due to inhomogenous local rotations.}
			a) Side view of two $8$-chain RVEs side-by-side along the $\euclid{1}$ direction (reference configuration).
			b) A gap, $\gapvec$, forms between $\rvec_1$ and $\rvec_2$ after inhomogeneous local rotation where the cross-link at $\xvec$ (left RVE) rotates by $\genRotStar$ and its neighboring cross-link (right RVE) rotates by $\genRotStar\left(\xvec\right) + \left(2 \Rmag / \sqrt{3}\right) \left(\takeGrad{\genRotStar}\right) \euclid{1}$; for certain cases, $\lVert \takeGrad{\F} \rVert \Rmag$ ``small'' implies that the magnitude of the gap, $\left|\gapvec\right| / \Rmag$, is also ``small''.
		}
	}  
	\label{fig:compatibility}
\end{figure}

Motivated by the above discussion, assume that $\lVert \takeGrad{\genRotStar} \rVert \Rmag \sim \lVert \takeGrad{\F} \rVert \Rmag$ and $\smallpar = \lVert \takeGrad{\genRotStar} \rVert \Rmag \ll 1$ (and $\lVert \Gradd^n \: \genRotStar \rVert \Rmag^n = \orderOf{\smallpar^2}$).
Assume further that $\genRotStar$ is a differentiable function of $\xvec$.
Then
	\begin{equation} \label{eq:smallgrad}
		\genRotStar\left(\xvec + \frac{2 \Rmag}{\sqrt{3}} \euclid{1}\right) = \genRotStar\left(\xvec\right) + \frac{2 \Rmag}{\sqrt{3}} \left(\takeGrad{\genRotStar}\right) \euclid{1} + \orderOf{\smallpar^2}.
	\end{equation}
Upon local rotation of the RVEs (and neglecting higher order terms, $\orderOf{\smallpar^2}$):
\begin{subequations}
\begin{align}
	\rvec_1 &= \genRotStar \Rvec_1 \\
	\rvec_2 &= \left(\genRotStar + \frac{2 \Rmag}{\sqrt{3}} \left(\takeGrad{\genRotStar}\right) \euclid{1} \right) \Rvec_2.
\end{align}
\end{subequations}
To satisfy compatibility (i.e. to ensure $\rvec_1$ and $\rvec_2$ still meet at a cross-link), we require that
\begin{equation} \label{eq:compat-1}
	\xvec + \rvec_1	= \xvec + \frac{2 \Rmag}{\sqrt{3}} \genRotStar \euclid{1} + \rvec_2.
\end{equation}
The condition \eqref{eq:compat-1} is not satisfied for general $\takeGrad{\genRotStar}$.
Its residual, the gap vector between $\rvec_1$ and $\rvec_2$, is
\begin{equation}
	\gapvec = \left(\xvec + \rvec_1\right) - \left(\xvec + \frac{2 \Rmag}{\sqrt{3}} \genRotStar \euclid{1} + \rvec_2\right) = -\frac{2 \Rmag}{\sqrt{3}} \left(\left(\takeGrad{\genRotStar}\right) \euclid{1}\right) \Rvec_2.
\end{equation}
The gap vector, $\gapvec$, is then mapped under $\F$ into the deformed configuration.
The magnitude of this gap relative to $\Rmag$ provides us with a measure of incompatibility:
\begin{equation}
	\frac{\sqrt{\gapvec \cdot \cGreen \gapvec}}{\Rmag} \leq \frac{2}{\sqrt{3}} \lVert \cGreen \rVert \smallpar.
\end{equation}
Depending on the mechanical properties of the polymer chains within the network (e.g. $\pLen / \cLen$ where, recall, $\pLen$ is the persistence length), it is reasonable to assume that the small gap, $\F \gapvec$, can be absorbed into the local chain confirmations with a negligible change in free energy density.
This motivates the validity of the free rotation assumption for some physical systems of interest.

The above analysis, however, is not general.
A more rigorous and comprehensive analysis of how $\netFreeEnergyDensity$ depends on inhomogeneous local rotation fields, $\genRotStar\left(\xvec\right)$, presents an interesting topic for future work.
We note, parenthetically, that the free rotation assumption may have some similarities to continuum field theory models for planar kirigami~\cite{zheng2022continuum}\footnote{\hl{Kirigami is similar to origami but with holes. Kirigami is sometimes called ``the art of paper cutting''--analogous to origami also being known as ``the art of paper folding''}}.
Roughly speaking, the repeat unit for planar kirigami consists of $4$ stiff quadrilateral plates and $4$ parallelogram slits that are connected to their nearest neighbors by hinges.
Theoretically, when the plates are rigid and the hinges are ideal, the rotation of a single repeat unit determines the deformation of the entire bulk metamaterial; that is, the deformation is a \emph{global mechanism}.
However, in practice, deformations tend to be more localized and consist of soft modes.
In fact,~\citet{zheng2022continuum} showed that, in the limit of finely patterned kirigami each repeat unit can be approximated by a deformation which is ``locally mechanistic''; that is, the effective deformation gradient can be decomposed into a field of local rotations and distortions.
We speculate that making a possible connection to locally rotating polymer networks more precise, or an analysis of locally rotating polymer networks along the lines of~\cite{zheng2022continuum} or~\cite{friesecke2002validity} may prove fruitful.
\end{hlbreakable}

\bibliographystyle{unsrtnat}
\bibliography{master}

\end{document}